\RequirePackage{fix-cm}
\documentclass[smallextended]{svjour3}       
\usepackage[utf8]{inputenc}
\usepackage{color}
\usepackage[dvipsnames]{xcolor}
\usepackage{float}
\usepackage{amsmath}
\usepackage{amssymb}
\usepackage{esint}
\usepackage{hyperref}
\usepackage{varioref}
\usepackage{seqsplit}
\smartqed 
\usepackage{slashbox}
\usepackage{algorithm}
\usepackage{graphicx}
\usepackage{subcaption}
\usepackage{scalerel}

\usepackage{algorithmic}

\def\Rn{\mathbb{R}^n}
\def\dx{\dot{x}}
\def\Rm{\mathbb{R}^m}
\def\Rr{\mathbb{R}}
\def\Cc{\mathbb{C}}

\def\Koop{\mathcal{K}}
\def\buc{\bar{u}(\cdot)}
\def\bu{\bar{u}}
\def\mcu{\mathcal{U}}
\def\fun{\mathcal{F}}
\def\mcf{\mathcal{F}}

\def\dataip{\{x_k,u_k,y_k\}_{k = 1}^K}

\def\DataNet{\{X_k,U_k,Y_k\}_{k=1}^K}
\def\DataNetAug{\{\bar{X}_k,\bar{Y}_k\}_{k=1}^K}
\def\bA{\bar{A}}
\def\bB{\bar{B}}

\def\bx{\bar{x}}
\def\by{\bar{y}}
\def\bX{\bar{X}}
\def\bY{\bar{Y}}
\def\dbX{\dot{\bar{X}}}
\def\hF{\hat{F}}
\def\hPhi{\bar{\Phi}}

\def\Uhts{\hat{\mathcal{K}}^{T_s}}

\def\bUk{A_K}
\def\mcX{\mathcal{X}}
\def\tL{L_K}

\def\te{\tilde{e}}
\def\tF{\tilde{F}}
\def\mbu{\mathbf{U}}
\def\mbp{\Pi}
\def\mio{\mathit{O}}
\def\dxi{\dot{x}_i}
\def\Nei{\mathcal{N}_i}
\def\Mei{\mathcal{M}_i}
\def\bnei{\bar{\mathcal{N}_i}}
\def\bmei{\bar{\mathcal{M}_i}}

\def\tba{\textbf{a}}
\def\tbb{\textbf{b}}
\def\tg{\tilde{g}}
\def\bq{\textbf{q}}
\def\br{\textbf{r}}
\newcommand\norm[1]{\left\lVert#1\right\rVert}
\def\mcx{\mathcal{X}}

\def\HXU{H_{X,U}}
\def\tphi{\boldsymbol{\phi}}
\def\tpsi{\boldsymbol{\psi}}

\def\bE{\bar{E}}
\begin{document}

\title{Koopman operator based identification of nonlinear networks\thanks{This work was done in full during R Anantharaman's tenure as a postdoctoral researcher at Department of Mathematics and naXys, University of Namur. He was supported by the F.R.S.-FNRS grant PDR-T.0148.22}
}


\author{Ramachandran Anantharaman         \and
        Alexandre Mauroy 
}


\institute{R Anantharaman \at
              Department of Electrical Engineering \\ Eindhoven University of Technology, Netherlands.\\
              \email{r.chittur.anantharaman@tue.nl}           
           \and
           Alexandre Mauroy \at
              Department of Mathematics and Namur Institute for Complex Systems (naXys),\\ University of Namur, Belgium. \\ \email{alexandre.mauroy@unamur.be}
}

\date{Received: date / Accepted: date}

\maketitle

\begin{abstract}
In this work, we develop a method to identify continuous-time nonlinear networked dynamics via the Koopman operator framework. The proposed technique consists of two steps: the first step identifies the neighbors of each node, and the second step identifies the local dynamics at each node from a predefined set of dictionary functions. The technique can be used to either identify the Boolean network of interactions (first step) or to solve the complete network identification problem that amounts to estimating the local node dynamics and the nature of the node interactions (first and second steps). Under a sparsity assumption, the data required to identify the complete network dynamics is significantly less than the total number of dictionary functions describing the dynamics. This makes the proposed approach attractive for identifying large dimensional networks with sparse interconnections. The accuracy and performance of the proposed identification technique are demonstrated with several examples.
\end{abstract}

\section{Introduction}

Dynamics defined over a network are useful in the mathematical modeling of electrical grids \cite{Nard_PG2014}, UAV networks \cite{Huang_UAV2021}, vehicular traffic networks \cite{Qu_VN2008,WM_VN2009}, and biological networks  \cite{TLip_MatMotGRN2004,JT_BioNet2019}, to name a few (see also the survey \cite{brugere2018network}). Such dynamics are generally described through a directed graph where each node represents a network component whose local dynamics depends on its own state, on the inputs from the neighboring network components, and on external inputs. Moreover, data-driven network modeling and, in particular, the reverse engineering problem of \textit{network identification} are of specific interest to physics, biology, and engineering applications. Network identification has been extensively studied in different settings. In the most classical setting, the network dynamics (i.e. local dynamics and interactions) are assumed to be known and the network topology is inferred from measurements of the node dynamics. In its simplest form, this problem amounts at constructing a Boolean graph that indicates the existence or absence of an edge between each pair of nodes. A great deal of methods have been proposed in this context, using various techniques from statistics and information theory \cite{brugere2018network,net_ident_info_theory,runge2018causal}, control theory \cite{Materassi2010,Yu_net_estimation}, optimization \cite{net_ident_opti,julius2009genetic}, linearization \cite{Sauer_net_ident}, compressed sensing \cite{net_ident_compressed_sens}, Bayesian inference \cite{chiuso2012bayesian}, etc. See also the introductory review paper \cite{Timme_review} and references therein. Alternative relaxed versions of the network reconstruction problems have also been considered, such as the spectral network identification problem aiming at recovering the eigenvalues of the graph Laplacian \cite{Franceschelli,MH_NetID2017}. In contrast to the above setting, another formulation of the network identification problem assumes that the network topology is known, while the dynamics and the interaction between nodes are unknown. This problem has been tackled in the linear case with unobserved nodes and external inputs, where global transfer functions between inputs and observed nodes are used to identify local transfer functions between neighboring nodes \cite{dankers2015identification,materassi2015identification}. The property of identifiability has also been considered in this context \cite{HGB_NetID2018,HMW_NetID2018}, and more recently in the nonlinear case \cite{vizuete2024nonlinear}.

The many problem settings in network reconstruction have their practical motivations based on constraints on available data and on the a priori knowledge on the network. In some practical scenarios, complete reconstruction of the network is essential. It implies not only inferring the whole network structure, but also identifying the local dynamics and the nature of the interactions, which may be different from one node to another in the case of heterogeneous networks. For linear dynamics, this problem is related to linear parameter estimation or state-space identification \cite{moonen1989and}, for which there exist many solutions and methods (e.g. \cite{GW_NSC2008}). Yet it can be seen from \cite{GW_NSC2008,PCW_NecFSM2013} that under no a priori information on the network, it is necessary to measure all the nodes for the complete reconstruction of the network. In the nonlinear case, this problem is undoubtedly much more involved and the tools developed for linear network identification are not immediately applicable for their nonlinear counterparts, requiring new investigations and algorithms for their identification.

A promising approach to nonlinear network identification is the so-called Koopman operator framework. The Koopman operator \cite{Koopman1} was introduced in the early 1930s to study nonlinear dynamical systems through a linear description of the evolution of observables defined on the state space. This approach received a renewed interest in analysis and control of nonlinear systems after the seminal work \cite{Mezic2005} on the spectral properties of the operator and the development of data driven techniques such as Dynamic Mode Decomposition \cite{Schmid2010} to capture those properties. This body of work lead to numerous data-driven approximations of the Koopman operator, thereby providing linear representations of nonlinear dynamics that can be leveraged in the context of nonlinear network identification. A first step in this direction was made in \cite{MG2020}, where a nonlinear identification problem was recast into a linear identification problem over a lifted dynamics. This approach yielded two methods relying on a computationally attractive linear algebraic formulation. However, both methods developed in \cite{MG2020} are not tailored to large-scale networked systems, not only requiring too many data points and lifted states, but also loosing accuracy as the system dimension grows. More recently, the network identification problem was also considered through the lenses of the Koopman operator in \cite{MO_KoopNetID22}, in a discrete-time context, and under the assumption of homogeneous local dynamics and identical coupling functions.

In this paper, we address the problem of complete network identification in a nonlinear setting through the Koopman operator framework. We consider the case of general nonlinear dynamics, with heterogeneous local dynamics and external inputs on the nodes, and different coupling functions, thereby relaxing the assumptions of \cite{MO_KoopNetID22} in the more involved continuous-time setting. Toward this aim, we leverage the methods developed in \cite{MG2020}, and overcome their limitations which currently impede their use in network identification. In particular, we extend the so-called dual method to non-autonomous systems, design a technique to optimize the choice of test functions, and more importantly develop a two-step method that fully identifies the whole network. In the first step, we use our improved dual method to identify the Boolean network of interactions, thereby identifying the neighbors of each node. Next, in the second step, we solve a local identification problem for each node among the (estimated) neighboring nodes. Under the realistic assumption of a sparse network, our two-step method is scalable (as evidenced by experiments with networks with up to $1500$ nodes), provides an accurate reconstruction of the Boolean network of interactions as well as better estimates of the local dynamics parameters (in comparison with the method developed in \cite{MG2020}). Moreover, as a by-product, this work leads to a modular linear representation of the networked dynamics, which could further be used for system analysis and control design.
Such modular representation is similar in essence to the one proposed in \cite{tellez2022data} in the context of predictive control and to a bilinear formulation recently proposed in \cite{guo2024modularized}.

The rest of the paper is organized as follows. An introduction to the Koopman operator framework and an overview of the identification techniques developed in \cite{MG2020} are provided in Section \ref{sec:prelim}. Section \ref{sec:DualMethod} extends the dual method of \cite{MG2020} to non-autonomous systems and provides a technique to optimize the choice of test functions. In Section \ref{sec:NetID}, we develop the two-step method for network identification. Numerical examples are shown in Section \ref{sec:NE}, which illustrate our methodology and its performance in the context of Boolean network reconstruction and identification of the dynamics. Finally, concluding remarks and perspectives are given in Section \ref{sec:conclu}.

\section{Koopman operator framework for dynamical systems}
\label{sec:prelim}

This section provides some preliminaries on Koopman operator, along with related data-driven techniques and an introduction to a specific application to system identification.

\subsection{The Koopman operator and its infinitesimal generator}

Consider a dynamical system
\begin{align}
    \label{eq:DS}
    \dot{x}(t) = F(x(t)),
\end{align}
where $x \in \Rr^n$ and $F : \Rr^n\to \Rr^n$ are the state and the vector field, respectively. The solution to \eqref{eq:DS} from an initial condition $x_0$ is given through the flow $\Phi^t$ by
\[
x(t) = \Phi^t(x_0).
\]
For a given space of functions $\mcf: \Rr^n \to \Cc$, the Koopman operator semigroup $\Koop^t$ \cite{Koopman1} associated with the system (\ref{eq:DS}) is defined by
\begin{align}
    \label{eq:Koop}
    \Koop^t h(x) = h(\Phi^t(x)),
\end{align}
for all $h \in \mcf$. The Koopman operator provides an alternative viewpoint to study (\ref{eq:DS}) as it is linear over $\mcf$ irrespective of the nonlinearity of the dynamics (\ref{eq:DS}). This linearity property is traded with the dimensionality since the space $\mcf$ is infinite-dimensional for dynamics over $\Rr^n$. Under additional assumptions on the dynamics and the function space $\mcf$, the Koopman operator $\Koop^t$ is a strongly continuous semigroup (${C}_0$-semigroup) in the space $\mcf$ \cite{AM-IM-YS2020}, a property which allows to define an infinitesimal generator $L$ as
\[
L h := \lim_{t \downarrow 0} \frac{\Koop^t h - h}{t},
\]
for all $h$ in the domain of $L$. In this work, unless explicitly mentioned, the space $\fun$ is the space of $L^2$ functions defined over the state space $\mathbb{R}^n$.

The Koopman operator can be extended to non-autonomous systems in the following way. Consider the non-autonomous system 
\begin{align}
    \label{eq:DS_ip}
    \dx(t) = F(x(t),u(t)),
\end{align}
where $x \in \Rn$ and $u \in \Rm$ are the state and inputs respectively. Let $\mcu$ denote the space of admissible control signals $\buc: \Rr_+ \to \Rm$. Given an initial condition $x_0$ and an input signal $\bu \in \mcu$, the solution to (\ref{eq:DS_ip}) is given through the flow
\[
x(t) = \Phi^{t}(x_0,u = \bu).
\]
Consider the space of functions $\fun: \Rn \times \mcu \to \Cc$, the Koopman operator $\Koop^t$ associated with (\ref{eq:DS_ip}) is defined as follows \cite{KM2018}
\begin{align}
    \label{eq:Koopman_ip}
    \Koop^t h(x,\buc) = h(\Phi^t(x,\bu),T^t(\buc))\ \ \ \forall \ \ h \in \fun,
\end{align}
where $T^t$ is the left shift operator defined as $T^t(\bar{u}(\tau)) := \bar{u}(t+\tau)$. 

\subsection{Data-driven approximation of the Koopman operator and lifted dynamics}
Practical algorithms using the Koopman operator rely on finite-dimensional approximations. They are typically data-driven, using snapshots of data $\{x_k,y_k\}_{k = 1}^K$ which satisfy (in the case of autonomous systems)
\[
y_k = \Phi^{T_s}(x_k),
\]
where $T_s$ is the sampling time. The most popular technique is the so-called Extended Dynamic Mode Decomposition (EDMD), which uses least-squares projections to construct a finite-dimensional approximation of the Koopman operator over a chosen subspace $\fun_N \subset \fun$. 

Given a set of linearly independent functions $\{\psi_i\}_{i = 1}^{N}$ that span the finite-dimensional subspace \textbf{$\fun_N$}, the aim of the EDMD algorithm is to construct
\[
\Koop_N := \Pi_N \Koop^{T_s}|_{\mcf_N},
\]
where $\Pi_N$ is the $L^2$ projection. The finite-dimensional operator $\Koop_N$ satisfies
\[
\Koop_N f = \underset{\tilde{f} \in \mcf_N}{\mathrm{argmin}}\norm{\Koop^{T_s}f - \tilde{f}}_2^2 \quad \forall \ f \in L^2(\Rr^n).
\]
Given the data points $\{x_k,y_k\}_{k=1}^{K}$, we can reformulate the above problem as
\begin{align}
\label{eq:EDMD}
\Koop_N f = \underset{\tilde{f} \in \mcf_N}{\mathrm{argmin}} \frac{1}{K} \sum_{k=1}^K \norm{\Koop^{T_s}f(x_k) - \tilde{f}(x_k)}_2^2 \quad \forall \ f \in L^2(\Rr^n),
\end{align}
and solve it as a least squares regression problem. In particular, the finite-dimensional operator $\Koop_N$ can be represented by the matrix
\[
A_N = P_y P_x^\dagger,
\]
where $P_x^{\dagger}$ is the pseudoinverse of $P_x$ and
\[
P_x = \begin{bmatrix} \Psi(x_1) & \Psi(x_2) & \dots & \Psi(x_K) \end{bmatrix} \quad P_y = \begin{bmatrix} \Psi(y_1) & \Psi(y_2) & \dots & \Psi(y_K) \end{bmatrix},
\]
with $\Psi = [\psi_1,\dots,\psi_N]^T$. The matrix $A_N$ allows to define the discrete-time \emph{lifted linear system}
\begin{align}
    \label{eq:LiftedDyn}
    z[k+1] = A_Nz[k],
\end{align}
where $z[k] = [\psi_1(x(kT_s)),\dots,\psi_{N}(x(kT(s))]^T$ is the lifted trajectory, which approximates the nonlinear continuous-time system \eqref{eq:DS}. A continuous-time counterpart for \eqref{eq:LiftedDyn} can be similarly constructed by approximating the infinitesimal generator $L_N$ over $\fun_N$. In practice, this can be obtained through the matrix logarithm of $A_N$, leading to the continuous-time linear lifted dynamics \cite{MG2020}
\[
\dot{z}(t) = L_Nz(t),
\]
where
\begin{equation}
\label{eq:Inf_Gen}
L_N = \frac{1}{T_s}\log(A_N).
\end{equation}

This data-driven approach has been extended to non-autonomous systems (\ref{eq:DS_ip}) in \cite{Proctor_DMD}, with the snapshots of data $\{x_k,u_k,y_k\}_{k=1}^K$ such that 
\[
y_k = \Phi^{T_s}(x_k,u = u_k),
\]
where the input $u$ is kept constant at $u_k$. Choosing the lifting functions 
\[
\{\psi_1(x),\dots,\psi_{N}(x) \} \cup \{\phi_1(u),\dots,\phi_m(u)\}
\]
where $\phi_i(u) = u_i$ are the lifting functions on the inputs, we obtain the discrete-time linear lifted dynamics
\begin{align}
\label{eq:LiftDyn_ip}
z[k+1] \approx Az[k] + Bu[k],
\end{align}
with the lifted states $z[k] = [\psi_1(x(kT_s)),\dots,\psi_{N}(x(kT_s))]$ and lifted inputs $u[k] = [\phi_1(u(k)),\dots,\phi_m(u(k))]^T$. 

\subsection{System identification using the Koopman operator framework}
\label{sec:syst_ident}

The work \cite{MG2020} presents a framework for system identification through the Koopman operator, which we recall in this subsection for the sake of completion. Assume that the vector field $F$ is of the form 
    \[
    F = \sum_{k=1}^{N_F} w_k h_k,
    \]
    where $h_k : \Rn \to \Rr$ are known dictionary functions and $w_k \in \Rn$ are unknown weighting coefficients. In this case, identifying the map $F$ from data is equivalent to estimating the weights $w_k$. Two approaches to this identification problem have been formulated in the Koopman operator framework, which fundamentally differ in their choice of basis functions. They are described as follows. 
\begin{enumerate}
    \item \textbf{Main method}: The basis functions are chosen to contain $\{h_k(x)\}_{k=1}^{N_f} \cup \{ id_j(x) \}_{j=1}^{n} \}$, where $id_j(x):= x_j$ yields the $j^{th}$ coordinate of $x$. The Koopman semigroup $\Koop^{T_s}$ is approximated by the EDMD method in the subspace spanned by the basis functions as the matrix $A_N$. This matrix is used to compute the matrix approximation $L_N$ of the infinitesimal generator as in (\ref{eq:Inf_Gen}). 
    Finally, the estimated coefficients $w^j_k$ of the vector field component $F_j$ lie in the column of $L_N$ associated with the basis functions $id_j$. Note that the total number $K$ of data pairs must satisfy $K> N$, where $N = N_F+n$. 
    As the total number of dictionary functions increases, more data points are required, making the method not suitable for large dimensional systems.  This motivates the alternative approach, termed as the dual method.

    \item \textbf{Dual Method}: The dual method aims to compute a matrix approximation of the adjoint of the Koopman operator $(\Koop^{T_s})^*$ in a ``sample space'' of dimension $K$. More precisely, the method seeks for the best approximation $A_K$ of the Koopman operator on an $M$-dimensional subspace such that, for any \textit{test functions} $g \in L^2(\Rr^n)$, we have
    \begin{equation}
    \label{eq:dual_AK}
         \begin{bmatrix} \Koop^{T_s} g(x_1) \\ \Koop^{T_s} g(x_2) \\ \vdots \\ \Koop^{T_s} g(x_K) \end{bmatrix}
    \approx A_K \begin{bmatrix} g(x_1) \\ g(x_2) \\ \vdots \\ g(x_K) \end{bmatrix}.
    \end{equation}
    Next, the matrix approximation $L_K$ of the infinitesimal generator is computed as in (\ref{eq:Inf_Gen}).
    This approximation is then used to compute the estimate of the vector field $\hat{F}$ at the data points $\{x_k\}_{k = 1}^K$, which is obtained as
    \begin{equation}
    \label{eq:VF_dual}
    \begin{bmatrix}
        \hat{F}(x_1)^T \\ \hat{F}(x_2)^T \\ \vdots \\ \hat{F}(x_K)^T 
    \end{bmatrix} = L_K \begin{bmatrix}
        x_1^T \\ x_2^T \\ \vdots \\ x_K^T
    \end{bmatrix}.
    \end{equation}
    Finally, the weights $w_k$ are estimated through an independent regression problem. 
\end{enumerate}

In the context of network identification, the amount of data available might be much less than the total number of functions to be identified in the network. To tackle this issue, we will combine both main and dual methods, reconstructing the network topology with a modified formulation of the dual method and subsequently using the main method to identify local dynamics.

\section{Dual lifting method for systems with inputs}

\label{sec:DualMethod}

While the dual identification method was initially developed for autonomous systems, it is extended to non-autonomous systems in this section. This extension is essential in the context of network identification, where external inputs can potentially act on the nodes. In addition, we provide a procedure to choose an appropriate set of test functions that minimizes the error in the estimation of the vector field.

\subsection{Lifting and construction of the Koopman Operator}

We are given samples $\dataip$ generated by the dynamics with a constant input over a sampling period, that is,
\[
y_k =  \Phi^{T_s}(x_k,u_k),
\]
where $\Phi^{T_s}$ is the flow associated with (\ref{eq:DS_ip}). The key idea is to use an augmented dynamics of the form 
\begin{align}
\label{eq:DS_ext_state}
    \dot{\bx} = \begin{bmatrix} \dot{x}(t) \\ \dot{u}(t) \end{bmatrix} = \begin{bmatrix}
        F(x,u) \\ 0
    \end{bmatrix} =: \hF(\bx),
\end{align}
where $\bx = [x,u]^T \in \mcX \subset \Rr^{n+m}$ is the augmented state. Accordingly, we define 
\begin{align}
\label{eq:dataip_concat}
\bx_k = [x_k,u_k]^T, \ \by_k = [y_k,u_k]^T
\end{align}
and it can be seen that $\by_k = \bar{\Phi}^{T_s} \bx_k$ where $\bar{\Phi}^{T_s}$ is the flow associated with the map $\hF$. The Koopman operator associated with this dynamics (\ref{eq:DS_ext_state}) is defined by
\[
\hat{\Koop}^{t} g(\bx) = g \circ \hPhi^t (\bx),
\]
where $g \in L^2(\mcX)$.  
In this case, similarly to \eqref{eq:dual_AK}, the dual method amounts to finding the matrix $\bUk$ that satisfies
\begin{align}
\label{eq:PrimalDual}
\begin{aligned}
    \begin{bmatrix}
        \Uhts g(\bx_1) \\ \vdots \\ \Uhts g(\bx_K) 
    \end{bmatrix} = \bUk \begin{bmatrix} g(\bx_1) \\ \vdots \\ g(\bx_K) \end{bmatrix}.
\end{aligned}
\end{align}
Given the data 
$\{\bar{x}_k,\bar{y}_k\}_{k=1}^K$, we construct the matrices
\begin{align}
\label{eq:liftedMat_Dual}
P_x = \begin{bmatrix} g(\bx_1)^T \\ g(\bx_2)^T \\ \vdots \\ g(\bx_K)^T \end{bmatrix}  \quad \quad P_y = \begin{bmatrix} g(\by_1)^T \\ g(\by_2)^T \\ \vdots \\ g(\by_K)^T 
\end{bmatrix},
\end{align}
where $g(\bx)$ is the vector of test functions $[g_1(\bx),\dots,g_N(\bx)]^T$. The total number $N$ of test functions needs to be equal to or greater than the number $K$ of samples in the data set, since each test functions acts as a data point in the optimization problem given in \eqref{eq:PrimalDual}. Note also that the input $u_k$ appears in both matrices $P_x$ and $P_y$. Next, the matrix $\bUk$ is given by
\begin{align}
\label{eq:UK_computation}
\bUk = P_y P_x^{\dagger}.
\end{align}
and we can finally compute an approximation of the vector field $\hF(\bx)$ at the data points $\{\bx_k\}_{k=1}^{K}$. For any $g \in L^2(\mcX)$, we have
\begin{align*}
\begin{bmatrix}
\hF(\bx_1)\cdot \nabla(g(\bx_1)) \\ \hF(\bx_2) \cdot \nabla(g(\bx_2)) \\ \vdots \\ \hF(\bx_K) \cdot \nabla(g(\bx_K))  
    \end{bmatrix} = \begin{bmatrix}
        L g(\bx_1) \\ L g(\bx_2) \\ \vdots \\ Lg(\bx_L) 
    \end{bmatrix} \approx \tL \begin{bmatrix} g(\bx_1) \\ g(\bx_2) \\ \vdots \\ g(\bx_K)\end{bmatrix},
\end{align*}
where $\tL$ is computed as 
\begin{align}
\label{eq:tL}
\tL  = \frac{1}{T_s} \log(\bUk).
\end{align}
 This equation can be used with $g(x)=x^T$ and, similarly to \eqref{eq:VF_dual}, the approximation of the vector field is obtained as 
\begin{align}
\label{eq:State_approx_dual}
\begin{bmatrix} \hat{F}(x_1,u_1)^T \\ \hat{F}(x_2,u_2)^T \\ \vdots \\ \hat{F}(x_K,u_K)^T   \end{bmatrix}  = \tL \begin{bmatrix} x_1^T \\ x_2^T \\ \vdots \\ x_K^T\end{bmatrix}.
\end{align}
Note that the input data points $u_k$ are not explicitly used in the above equality. In fact, the effect of the inputs is implicitly captured in the matrix $\tL$, and the estimates of the vector field are valid only for the specific state-input pairs $\{x_k,u_k\}_{k=1}^{K}$.

\subsection{Test functions for computing the Koopman Operator}
\label{sec:TestFun}

Unlike the main method, where the lifting functions are selected according to some prior knowledge of the dictionary functions, there is no constraint on the choice of test functions $g$ in the dual method \eqref{eq:PrimalDual}. The set of test functions can therefore be chosen to optimize the quality of the approximation of the Koopman operator. Though this problem can be posed in multiple ways, this section focuses on the following specific formulation. We are given several sets $G_i = \{g_{i1},\dots,g_{iN} \}$ of test functions. Our aim is to select the best set $G_i$ so that the error in the estimation of the vector field $F(x,u)$ is minimum. For each choice of test functions $G_i$, we denote by $A_i$ the corresponding approximation (\ref{eq:UK_computation}) of the Koopman operator and we define the $K\times n$ approximation error matrix
\begin{small}
\begin{align}
\label{eq:e_i}
e_i = \begin{bmatrix} y_1^T \\ \vdots \\ y_K^T \end{bmatrix} - A_i \begin{bmatrix} x_1^T \\ \vdots \\ x_K^T \end{bmatrix}.
\end{align}
\end{small}
Let us also define $L_i = \frac{1}{T_s} \log(A_i)$ and the error
\begin{small}
\begin{align}
\label{eq:te_i}
\te_i = \begin{bmatrix} F(x_1,u_1)^T \\ \vdots \\ F(x_K,u_k)^T \end{bmatrix} - L_i \begin{bmatrix} x_1^T \\ \vdots \\ x_K^T \end{bmatrix}.
\end{align}
\end{small}
The following result characterizes the best choice of the test functions.
\begin{proposition} 
\label{lem:Test-fun}
Given the sets of test functions $G_i$ and the associated errors \eqref{eq:e_i}-\eqref{eq:te_i}, with $i=1,\dots,r$,
\[
\lim_{T_s \to 0} \underset{i \in \{1,\dots,r\}}{\mathrm{argmin}} \|e_i\| = \lim_{T_s \to 0}  \underset{i \in \{1,\dots,r\}}{\mathrm{argmin}} \|\te_i\|.
\]

\end{proposition}
\begin{proof}
Define the data matrices
\[ 
\bY = \begin{bmatrix} y_1^T \\ \vdots \\ y_K^T \end{bmatrix} \quad \bX = \begin{bmatrix} x_1^T \\ \vdots \\ x_K^T \end{bmatrix} \quad \dbX = \begin{bmatrix} F(x_1,u_1)^T \\ \vdots \\ F(x_K,u_k)^T \end{bmatrix},
\]
and suppose that $i^*=\mathrm{argmin}_i \|e_i\|$.
Since $A_{i^*}$ minimizes $\|\bY - A_i \bX\|$, we have
\begin{align*}  
   \| \bY - e^{L_{i^*} T_s} \bX \| & \leq \| \bY - e^{L_i T_s} \bX \| \quad \forall i.  
\end{align*}
Expanding $e^{At}$ and discarding terms of the order $\mio(T_s^2)$, we obtain
\[
   \norm{T_s \bigg( \frac{\bY - \bX}{T_s} - L_{i^*} \bX \bigg)}  \leq \norm{ T_s \bigg( \frac{\bY - \bX}{T_s} - L_i \bX \bigg) }, 
   \]
or equivalently
 \[
  \norm{\bigg( \frac{\bY - \bX}{T_s} - L_{i^*} \bX \bigg)}  \leq \norm{ \bigg( \frac{\bY - \bX}{T_s} - L_i \bX \bigg) }.
 \] 
In the limit $T_s \to 0$, the above equality yields
\[
\| \dbX - L_{i^*} X\| \leq \| \dbX - L_i X\|  \quad \quad \forall i,
\]
so that $i^*=\mathrm{argmin}_i \|\te_i\|$, which concludes the proof.
\qed
\end{proof}
 
 The above result gives a quantitative way to select the best test functions. Given the fact that the dual method approximates the Koopman operator over the data points for arbitrary functions, the work \cite{MG2020} suggests to use Gaussian radial basis functions (RBF) centered around the data points, that is,
\begin{align}
\label{eq:RBF}
 g_k(\bx) =  e^{-\gamma^2\|\bx-\bx_k\|^2} \quad \quad   1\leq k \leq K,
\end{align}
where $\gamma \geq 0$ is a parameter. Given several sets of Gaussian RBFs parameterized by different values $\gamma$, it follows from Proposition \ref{lem:Test-fun} that the set of RBF which minimizes the error in equation \eqref{eq:e_i} is also expected to provide the smallest error in the approximation of the vector field. This is illustrated with the following example.

 \paragraph{Example.}
 Consider the dynamics
   \begin{align*}
       \begin{aligned}
           \dot{x}_1 &= 2x_2 + 3x_3^2 \\
           \dot{x}_2 &=-0.8x_1x_3-2x_1^2x_2 + u \\
           \dot{x}_3 &= -x_2x_3 + x_1u^2.
       \end{aligned}
   \end{align*}
We use $50$ different samples of initial conditions and inputs, uniformly distributed over $[-1,1]^3$ and $ [-1,1]$, respectively, and we generate data pairs with a sampling time $T_s = 0.01$. The dual Koopman operator representation is computed for $20$ different sets of $50$ different Gaussian RBFs centered around the data points, with $\gamma$ varying from $0.0025$ to $0.025$. As shown in Figure \ref{fig:main}, the value $\gamma$ which minimizes the error $e$ in the approximation of the dual Koopman operator (see \eqref{eq:e_i}) also minimizes the error in $\tilde{e}$ in the computation of the vector field (see \eqref{eq:te_i}).

\begin{figure}[ht]
\centering
\includegraphics[scale=.6]{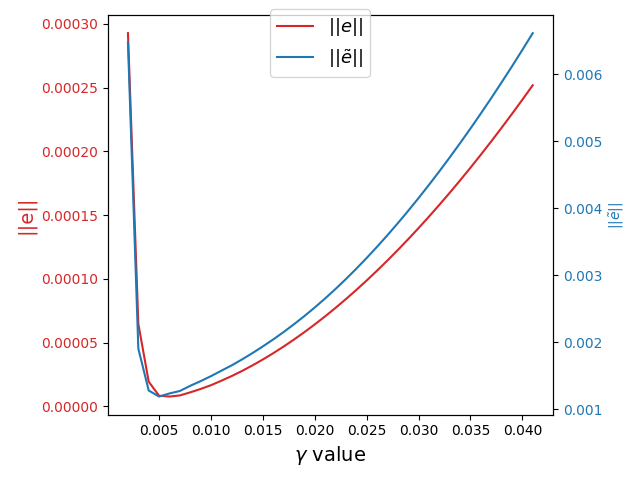}
\caption{According to Proposition \ref{lem:Test-fun}, the errors $e$ and $\tilde{e}$ reach their minimum at the same value of the parameter $\gamma$.}
\label{fig:main}
\end{figure}


\section{Network Identification using Koopman Operator}
\label{sec:NetID}

In this section, we aim to identify a nonlinear networked dynamics through the Koopman operator framework. Specifically, we will develop a two-step procedure where we perform the identification of the topology in the first step and the local identification of the network dynamics in the second step. When the network is sparse, our identification algorithm requires much less data than the total number of parameters to be identified, making it tractable for large networks.

\subsection{Problem formulation}
Consider a network of $N$ nodes and $M$ inputs given by the dynamics
\begin{align}
    \begin{aligned}
    \dxi = F_{i}(x_i) + \sum_{k \in \Mei} G_{ik}(u_k) + \sum_{k \in \Nei} H_{ik}(x_k),\\ 
    \end{aligned}
    \label{eq:NetDyn}
\end{align}
where $x_i \in \Rr^{n_i}$ is the state of the node $i$, $u_{k} \in \Rr^{m_j}$ is an input, $\Mei \subseteq \{1,\dots,M\}$ is the set of inputs at node $i$, and $\Nei \subseteq \{1,\dots,N\}$ is the set of neighbors for the node $i$. The functions $F_{i}: \Rr^{n_i} \to \Rr^{n_i}$, $G_{ik}: \Rr^{m_k} \to \Rr^{n_i}$ and $H_{ik}: \Rr^{n_k} \to \Rr^{n_i}$ describe the internal dynamics at the node $i$, the effect of the inputs on that node, and the coupling with its neighbor nodes, respectively. The network dynamics can be represented by
\begin{align}
\label{eq:NW}
\dot{X} = \begin{bmatrix} \bar{F}_1(X,U) \\ \vdots \\ \bar{F}_N(X,U) \end{bmatrix} = \bar{F}(X,U),
\end{align}
where $X = [x_1,\dots,x_N]^T \in \mcX \subset \Rr^n$ ($n = \sum_i n_i$) is the vector of states, $U = [u_1,\dots,u_N]^T \in \Rr^m$ ($m = \sum_i m_i$) is the vector of inputs, and $\bar{F}_i : \Rr^n \times \Rr^m \to \Rr^{n_i}$ is defined as in (\ref{eq:NetDyn}). Let $\mcu$ be the space of admissible input signals $\bar{U}(\cdot): \Rr_+\to \Rr^m$. 

Starting from the initial condition $X_0 = [x_1(0),\dots,x_N(0)]^T$ and an input signal $\bar{U} \in \mcu$, the solution $X(t)$ at time $t$ associated with the dynamics (\ref{eq:NW}) is defined through the flow as  
\begin{align*}
X(t) = \Phi^{t}(X_0,U = \bar{U}),
\end{align*}
and the Koopman operator associated with the network is defined as 
\begin{align}
\label{eq:Koop_Net}
    \Koop^t \psi(X,\bar{U}) = \psi(\Phi^t(X,U), T^t \bar{U}),
\end{align}
where $\psi \in \mathcal{F}: \Rr^n \times \mcu \to \Cc$, and $T^t$ is a right shift operator defined as 
\[
T^t \bar{U}(\tau) = \bar{U}(t+\tau).
\]
We assume that the functions $F_i$, $G_{ik}$, and $H_{ik}$ are of the form
\begin{align}
    \begin{aligned}
        F_{i}(x_i) &= \sum_{l=1}^{N_i} \alpha_{il} f_{il} (x_i),  \\
        G_{ik}(u_k) &= \sum_{l=1}^{M_k} \beta_{ikl} g_{kl} (u_k) \quad k \in \Mei,\\ 
        H_{ik}(x_k) &= \sum_{l=1}^{R_k} \gamma_{ikl} h_{kl} (x_k) \quad k \in \Nei,
    \end{aligned}
    \label{eq:VF_comp}
    \end{align}
where $f_{il}: \Rr^{n_i} \to \Rr$, $g_{kl}: \Rr^{m_k} \to \Rr$, $h_{kl}: \Rr^{n_k} \to \Rr$ are known dictionary functions and $\alpha_{il},\beta_{ikl}, \gamma_{ikl} \in \Rr^{n_i}$ are the parameters to be estimated. We further assume that the functions $f_{il}$, $g_{kl}$, and $h_{kl}$ belong to the space $L^2$.

For performing the identification, we assume that the data is available as $\{X_k,U_k,Y_k\}_{k=1}^K$ where $X_k = [x_{1k},\dots,x_{Nk}]$ is the initial condition, $Y_k = [Y_{1k},\dots,Y_{Nk}]^T$ is the flow at a known time $T_s$, with the external input $U_k = [u_{1k},\dots,u_{Nk}]^T$ held constant. In this setting, we develop a two-step network identification procedure with the following objectives.
\begin{enumerate}
    \item \textbf{First step - Identification of $\Nei$ and $\Mei$}: For each node, we identify a set of nodes $\bnei$ and the set of inputs $\bmei$ such that, under ideal conditions, these sets correspond to the sets $\Nei$ and $\Mei$, respectively. This is performed using the dual method developed in Section \ref{sec:DualMethod}.
    
    \item \textbf{Second step - Local identification of dynamics}: Once the estimates of the neighbors $\bnei$ and the inputs $\bmei$ are computed, a subset of dictionary functions can be obtained for each node and used to estimate the parameters $\alpha_{il}$, $\beta_{ikl}$ and $\gamma_{ikl}$. This is performed through the main method presented in Section \ref{sec:syst_ident}.
\end{enumerate}

\subsection{First step - Identification of the topology}

The first step of the network identification method consists of (1) estimating the vector field at each sample point and (2) identifying the neighbors $\Nei$ and inputs $\Mei$ acting on the node $i$.
Given the dynamics (\ref{eq:NetDyn}) and the data $\DataNet$, the dual method approach developed in Section \ref{sec:DualMethod} is used to compute an estimate of the time derivatives of $X_k$ by approximating the Koopman operator over the sample space by a suitable choice of test functions. We now proceed to the identification of the neighbors $\Nei$ and inputs $\Mei$ using these estimates of $\dot{X}_k$.

\subsubsection{Choice of node functions}

In order to identify the network topology, a direct approach consists in using the estimates of the vector field $\dot{X}$ and solving a regression problem over the complete set of dictionary functions of the network (see e.g. the Lasso regression problem in \cite{MG2020}). However, when the network has a large number of nodes, this would lead to solving a regression problem over a very large set of dictionary functions, affecting the identification error. Instead, we can solve a smaller regression problem for each node over a set of ``node functions", where the total number of such functions is much smaller than the total number of dictionary functions and still be able to identify all the edges in the network.

For each node $x_k$ and input $u_k$, we define two sets $\{\phi_{kl}\}_{l=1}^d$ and $\{\psi_{kl}\}_{l=1}^d$ of $d$ ``node functions" depending on $x_k$ and $u_k$, i.e. $\phi_{kl} : \Rr^{n_k} \to \Rr$ and $\psi_{kl} : \Rr^{m_k} \to \Rr$. Additionally, these functions are chosen to satisfy
\begin{align}
\label{eq:cond_node_fct_1}
\int_\mcx  \phi_{kl}(x_k) \, dx &= 0 \quad \forall k = 1,\dots,N \quad l=1,\dots,d. \\
\label{eq:cond_node_fct_2}
\int_\mbu \psi_{kl}(u_k) \, du &= 0 \quad  \forall k = 1,\dots,M \quad l=1,\dots,d.
\end{align}
We have the following useful result.
\begin{lemma}
\label{lem:Lem1}
Suppose that the two sets of functions $\{\phi_{kl}\}_{l=1}^d$ and $\{\psi_{kl}\}_{l=1}^d$ satisfy \eqref{eq:cond_node_fct_1} and \eqref{eq:cond_node_fct_2}. Given a function $\tg$ of the form
\[
\tg = \sum_{i = 1}^N \tg_{1i} + \sum_{i=1}^M \tg_{2i},
\]
where $\tg_{1i}$ and $\tg_{2i}$ depend on $x_i$ and $u_i$, respectively, we have

\begin{align*}
  &\tg_{1i} = 0 \implies \langle \tg,\phi_{ij} \rangle = 0 \quad \forall\  j = 1,\dots,d.  \\
  & \tg_{2i} = 0 \implies \langle \tg,\psi_{ij} \rangle = 0 \quad \forall\  j = 1,\dots,d.
\end{align*}

\end{lemma}

\begin{proof} For $i\neq k$, we have
\begin{align*}
    \langle \tg_{1k}, \phi_{ij} \rangle= |\mbu|\int_\mcx \tg_{1k}(x_k) \,dx \, \int_\mcx \phi_{ij}(x_i) \,dx= 0, \\
    \langle \tg_{2k}, \psi_{ij} \rangle= |\mcx| \int_{\mcu} \tg_{2k}(u_k) \,du \, \int_{\mcu} \psi_{ij}(u_i) \,du= 0,
\end{align*}
where $|\mcX|$ and $|\mbu|$ denotes the Lebesgue measure of $\mcX$ and $\mbu$ respectively and we used \eqref{eq:cond_node_fct_1} and \eqref{eq:cond_node_fct_2}.  Moreover, it is clear that
\begin{align*}
   \langle \tg_{1k}, \psi_{ij} \rangle= \int_\mcx \tg_{1k}(x_k) \,dx \, \int_\mbu \psi_{ij}(u_i) \,du = 0,\\
   \langle \tg_{2k}, \phi_{ij} \rangle= \int_\mbu \tg_{2k}(u_k) \,du \, \int_\mcx \phi_{ij}(x_i) \,dx= 0.
\end{align*}
Then, it follows that
\begin{align*}
 \langle \tg, \phi_{ij} \rangle &= \left\langle \sum_{k=1}^N \tg_{1k} + \sum_{k=1}^N \tg_{2k} , \phi_{ij} \right\rangle = \sum_{k=1}^{N} \langle \tg_{1k}, \phi_{ij}\rangle +\sum_{k=1}^{N} \langle \tg_{2k}, \phi_{ij}\rangle = \langle \tg_{1i}, \phi_{ij} \rangle,
\end{align*}
and similarly, we get $\langle \tg, \psi_{ij} \rangle =\langle \tg_{2i}, \psi_{ij} \rangle$. Finally, if $\tg_{1i} = 0$ or $\tg_{2i} = 0$, we obtain the required result.
\qed
\end{proof}

\subsubsection{Estimating the neighbors of a node}
\label{sec:estimate_neighbors}

We can now derive theoretical results that will allow us to identify the neighbors of a node.
Define the subspaces $\mcf_{x,k} = \mbox{span} \{ \phi_{kl} \}_{l=1}^d \subset L^2(\mcx,\mcu)$, $\mcf_{u,k} = \mbox{span} \{ \psi_{kl}\}_{l=1}^d \subset L^2(\mcx,\mcu)$, and $\mcf_{x,u} = \oplus_k \{\mcf_{x,k}\}_{k=1}^N \oplus \{ \mcf_{u,k} \}_{k = 1}^{M}$. We denote the $j^{th}$ component of the function $\bar{F}_i$ as $\bar{F}_i^{(j)}$. For each $i$, we will aim to compute
\begin{align}
\label{eq:MinF_ij}
\begin{aligned}
\tF_{i}^{(j)} = \underset{\tF_{i}^{(j)} \in \mcf_{x,u}}{\mathrm{argmin}} \ \ \norm{\bar{F}_{i}^{(j)} - \tF_{i}^{(j)}}^2 
&=  \int_{\mcx \times \mbu} (\bar{F}_{i}^{(j)}(X,U) - \tilde{F}_{i}^{(j)}(X,U))^2 dU dX. 
\end{aligned}
\end{align}
The function $\tF_{i}^{(j)}$ can be viewed as the orthogonal projection of $\bar{F}_{i}^{(j)}$ on to the subspace $\mcf_{x,u}$.
Thus, it can be written as 
\small{\begin{align}
\label{eq:TF_ij}
\tF_{i}^{(j)}(X,U) = \sum_{k = 1}^{N} \sum_{l = 1}^d a_{ikl}^{(j)} \phi_{kl}(x_k) + \sum_{k = 1}^M \sum_{l=1}^d b_{ikl}^{(j)}\psi_{kl}(u_k) = \sum_{k=1}^N (\tba_{ik}^{(j)})^T \tphi_{k} + \sum_{k=1}^M (\tbb_{ik}^{(j)})^T \tpsi_k,
\end{align}}
where $\tphi_k = [\phi_{k1},\dots,\phi_{kd}]^T$, $\tpsi_k= [\psi_{k1},\dots,\psi_{kd}]^T$, $\tba_{ik}^{(j)} = [a_{ik1}^{(j)},\dots,a_{ikd}^{(j)}]^T$, and $\tbb_{ik} = [b_{ik1}^{(j)},\dots,b_{ikd}^{(j)}]^T$. From the vectors $\tba_{ik}^{(j)}$ and $\tbb_{ik}^{(j)}$, we define
\begin{align}
\label{eq:NeiWeights}
\Lambda_{ik} = \sum_{j=1}^{n_i} \norm{\tba_{ik}^{(j)}}_1 \quad \quad \Delta_{ik} = \sum_{j=1}^{n_i} \norm{\tbb_{ik}^{(j)}}_1,
\end{align} 
where $\norm{ \cdot }_1$ is the $1$-norm. We will use these quantities to determine the existence of an edge between node $k$ and node $i$, and similarly to determine whether input $k$ affects node $i$.

\begin{theorem}

\label{th:neighbors}
Consider a network of $N$ nodes and $M$ inputs with dynamics \eqref{eq:NetDyn}, and functions $\phi_{kl} : \Rr^{n_k} \to \Rr$ and $\psi_{kl} : \Rr^{m_k} \to \Rr$ satisfying \eqref{eq:cond_node_fct_1}-\eqref{eq:cond_node_fct_2}. For each node $i$, the weights $\Lambda_{ik},\ k = 1,\dots,N$, and $\Delta_{ik}, \ k=1, \dots,M$, are computed as given by \eqref{eq:NeiWeights}. Then,
\[F_{i} = 0 \implies \Lambda_{ii} = 0, \quad \quad H_{ik} = 0 \implies \Lambda_{ik} = 0, \quad \quad
    G_{ik} = 0 \implies \Delta_{ik} = 0.\]
Moreover, if the components of $F_{i}$, $H_{ik}$ and $G_{ik}$ are not orthogonal to the spaces $\mcf_{x,i}$, $\mcf_{x,k}$ and $\mcf_{u,k}$, respectively, then 
    \[
    F_{i} = 0 \iff \Lambda_{ii} = 0, \quad \quad H_{ik} = 0 \iff \Lambda_{ik} = 0, \quad \quad G_{ik} = 0 \iff \Delta_{ik} = 0.
    \]
\end{theorem}

\begin{proof}

Let us define $\tF_{x,ik}^{(j)} = (\tba_{ik}^{(j)})^T \tphi_k \in \fun_{x,k}$ and $\tF_{u,ik}^{(j)} = (\tbb_{ik}^{(j)})^T \tpsi_k \in \fun_{u,k}$, so that we can write
\[
\tF_i^{(j)} = \sum_{k=1}^N \tF_{x,ik}^{(j)} + \sum_{k=1}^M \tF_{u,ik}^{(j)}.
\]
Since $\tF_i^{(j)}$ is computed from an orthogonal projection (\ref{eq:MinF_ij}) and from \eqref{eq:cond_node_fct_1}-\eqref{eq:cond_node_fct_2} it follows that $\fun_{x,k} \perp \fun_{x,l}$, $\fun_{u,k} \perp \fun_{u,l}$ for all $k\neq l$, and $\fun_{x,k} \perp \fun_{u,l}$ for all $k,l$. It appears that $\tF_{x,ik}^{(j)}$ and $\tF_{u,ik}^{(j)}$ are the projections of $\bar{F}_i^{(j)}$ onto the subspaces $\fun_{x,k}$ and $\fun_{u,k}$, respectively. This implies that
\begin{align}
\label{eq:temp2}
\begin{aligned}
0 = \langle \bar{F}_i^{(j)} - \tF_{x,ik}^{(j)}, \tF_{x,ik}^{(j)} \rangle &= \langle \bar{F}_i^{(j)}, \tF_{x,ik}^{^{(j)}} \rangle - \langle\tF_{x,ik}^{(j)}, \tF_{x,ik}^{(j)} \rangle \\
&= \langle \bar{F}_i^{(j)}, \tF_{x,ik}^{^{(j)}} \rangle - \|\tF_{x,ik}^{(j)}\|^2
\quad \quad \forall k = 1,\dots,N.
\end{aligned}
\end{align}

Moreover, it follows from \eqref{eq:NetDyn} and Lemma \ref{lem:Lem1} that
\begin{align*}
    \langle \bar{F}_i^{(j)}, \tF_{x,ik}^{^{(j)}} \rangle = \left\langle \left(F_{i}^{(j)} + \sum_{r \in \Mei} H_{ir}^{(j)} + \sum_{r \in \Mei} G_{ir}^{(j)}\right), \tF_{x,ik}^{(j)} \right \rangle =\langle H_{ik}^{(j)}, \tF_{x,ik}^{(j)}  \rangle,
\end{align*}
and with \eqref{eq:temp2}, we obtain
\begin{align}
    \label{eq:temp3}
    \langle H_{ik}^{(j)}, \tF_{x,ik}^{(j)}  \rangle =  \|\tF_{x,ik}^{(j)}\|^2.
\end{align}
We now prove the statements of the theorem. It follows from \eqref{eq:temp3} that 
\begin{align*}
H_{ik} = 0 \implies \langle H_{ik}^{(j)}, \tF_{x,ik}^{(j)}  \rangle = 0 \,\,  \forall j \implies \tF_{x,ik}^{(j)} = 0 \,\,  \forall j
\implies \tba_{ik}^{(j)} = 0^T \,\,  \forall j
\implies \Lambda_{ik} = 0.
\end{align*}
Additionally, if the components of $H_{ik}$ are not orthogonal to $\mcf_{x,k}$,
\begin{align*}
H_{ik} \neq 0 \implies \exists\ j \  \textrm{s.t.} \ \langle H_{ik}^{(j)}, \tF_{x,ik}^{(j)} \rangle \neq 0 
\implies \tF_{ik}^{(j)} \neq 0 
\implies \tba_{ik}^{(j)} \neq 0^T
\implies \Lambda_{ik} \neq 0.
\end{align*}
The proof for the components $F_{i}$ and $G_{ik}$ follows on similar lines and is omitted.
\qed
\end{proof}
Theorem \ref{th:neighbors} shows that we can use orthogonal projections and node functions to identify the neighbors of a node. In practice, we will estimate $\Lambda_{ik}$ and $\Delta_{ik}$ (see \eqref{eq:NeiWeights}) from the data and the estimates $\hat{F}_i(X,U)$ of the vector field. Given a choice of node functions $\phi_{il}, \psi_{il}$, we define 
\[
p_i(x_i) = [\phi_{i1}(x_i),\dots, \phi_{id}(x_i)]^T \quad \quad q_i(x_i) = [\psi_{i1}(u_i),\dots,\psi_{id}(u_i)]^T 
\]
and 
\[
p(X,U) = [p_1(x_1)^T, \dots, p_N(x_N)^T, q_1(u_1)^T, \dots, q_M(u_M)^T ]^T,
\]
and we construct the matrix $\HXU$ 
\begin{equation}
\label{eq:Hxu-network}
\HXU = \begin{bmatrix} p(X_1,U_1)^T \\ p(X_2,U_2)^T \\ \vdots \\ p(X_K,U_K)^T \end{bmatrix}.
\end{equation}
From the estimates of $\hat{F}_{i}^{(j)}$, we set up the optimization problem 
\[
{}^{argmin}_{\ \ \ \hat{\xi}_{i}^{(j)}} \norm{ \begin{bmatrix} \hat{F}_{i}^{(j)}(X_1,U_1) \\ \vdots \\ \hat{F}_{i}^{(j)}(X_K,U_K) \end{bmatrix} - \HXU \hat{\xi}_{i}^{(j)} }_2^2,
\]
where $\hat{\xi}_{i}^{(j)} \in \Rr^{(N+M)d}$ is the optimization variable $\hat{\xi}_{i}^{(j)} = [\tba_{i1}^{(j)}\, \cdots  \, \tba_{iN}^{(j)} \,  \tbb_{i1}^{(j)} \, \cdots \,  \tbb_{iM}^{(j)}]^T$.
If $K < (N+M)d$ or if we need to promote sparsity of the network, we can regularize the problem by adding a penalty term and solving the following Lasso regression problem 
\begin{align}
\label{eq:LassoNet}
{}^{argmin}_{\ \ \ \hat{\xi}_{i}^{(j)}} \norm{ \begin{bmatrix} \hat{F}_{i}^{(j)}(X_1,U_1) \\ \vdots \\ \hat{F}_{i}^{(j)}(X_K,U_K) \end{bmatrix} - \HXU \hat{\xi}_{i}^{(j)} }_2^2 + \rho \norm{ \hat{\xi}_{i}^{(j)}}_1.
\end{align}
Finally, the estimates of $\Lambda_{ik}$ and $\Delta_{ik}$ are computed from $\tba_{ik}^{(j)}$ and $\tbb_{ik}^{(j)}$ respectively by using (\ref{eq:NeiWeights}).
In practice, a threshold $\delta$ can be considered on $\Lambda_{ik}$ and $\Delta_{ik}$. In this case, the set of neighbors $\bar{\Nei}$ and the set of inputs $\bar{\Mei}$ are computed as
\begin{equation}
\label{eq:estim_neighbors}
\bar{\Nei} = \{k | \Lambda_{ik} \geq \delta\} \quad \quad \bar{\Mei} = \{k | \Delta_{ik} \geq \delta \}.
\end{equation}

\subsection{Second step - local identification of the dynamics}

In the first step of the method, we have obtained the estimates $\bnei$ and $\bmei$ of the sets of neighbors and inputs, respectively, for each node. In the second step, we estimate the functions $F_{i}, G_{ik}, H_{ik}$, assuming that they have the form (\ref{eq:VF_comp}) and depend on states and inputs related to the estimates $\bnei$ and $\bmei$. For this local identification, we compute a finite-dimensional ``local" approximation of the Koopman operator $\Koop^{T_s}$ from the data and use it to determine the parameters $\alpha_{il}$, $\beta_{kil}$ and $\gamma_{kil}$. For a given node $i$, the local identification involves two steps, namely:
\begin{enumerate}
    \item Lifting the data and computing the restricted Koopman operator over a finite-dimensional subspace of $\mcf$.
    \item Constructing an approximate continuous-time linear lifted dynamics and recovering the parameters. 
\end{enumerate}

\subsubsection{Lifting of the data and approximation of the Koopman operator}

For each node $i$, we define the neighbor state and input vectors

\begin{align}
\label{eq:nei_state_ip}
\mathbf{x}_i = \begin{bmatrix} x_{i_1} \\ x_{i_2} \\ \vdots \\ x_{i_{s_i}}\end{bmatrix} \quad \quad i_k \in \bnei
\qquad
\mathbf{u}_i = \begin{bmatrix} u_{i_1} \\ u_{i_2} \\  \vdots \\ u_{i_{t_i}} \end{bmatrix} \quad \quad i_k \in \bmei.
\end{align}
Let $s_i = |\bnei|$ and $t_i = |\bmei|$.
From \eqref{eq:VF_comp} define the dictionary functions vectors as
\begin{align}
 \label{eq:S2DictFun}
 \begin{aligned}
     p(x_i) &= [f_{i1}(x_i),\dots,f_{iN_i}(x_i)], \\
     q_k(x_k) &= [h_{k1}(x_k),\dots, h_{kR_{k}}(x_k)] \quad \quad \forall \  k \in \bnei, \\
     \bq_i(\mathbf{x}_i) &= [q_{i_1}(x_{i_1}), \dots, q_{i_{s_i}}(x_{i_{s_i}})], \\
     r_{k}(u_k) &= [g_{k1}(u_k),\dots, g_{kM_{k}}(u_k)] \quad \quad \forall \ k \in \bmei,\\
     \br_i(\mathbf{u}_i) &= [r_{i_1}(u_{i_1}), \dots, r_{i_{t_i}}(u_{i_{t_i}})].
 \end{aligned}
\end{align}
Let $\bar{R}_{i} = \sum_{i_k \in \bnei} R_{i_k}$ and $\bar{M}_i = \sum_{i_k \in \bmei} M_{i_k}$.
Using the dictionary functions, we define the subspaces
 \begin{align*}
 \mcf_{x_i} &= \mbox{span}\{f_{il}\}, \\ 
\mcf_{\mathbf{x}_i} &= \mathop{\oplus}_{\substack{k \in \bnei}} \{ h_{kl} \},\\
\mcf_{\mathbf{u}_i} &= \mathop{\oplus}_{\substack{k \in \bmei}}  \{ g_{kl}\},
\end{align*}
and
\begin{align}
\label{eq:BasFun_loc}
\mcf_i := \mcf_{x_i} \oplus \mcf_{\mathbf{x}_i} \oplus \mcf_{\mathbf{u}_i} \subset \mcf.
\end{align}
In the context of local identification, we will approximate the operator
\begin{align*}
    \Koop^{t}_i : \mcf_{x_i} \to \mcf_i, \qquad \Koop_i^{t} = \mbp_i \Koop^{t}|_{\mcf_{x_i}},
\end{align*}
where $\mbp_i$ is the $\mathcal{L}^2$ orthogonal projection $\mbp_i : \mcf \to \mcf_i$. 
Given $f \in \mcf_{x_i}$, we aim to compute $\bar{f} \in \mcf_{x_i}$, $\bar{h} \in \mcf_{\mathbf{x}_i}$, $\bar{g} \in \mcf_{\mathbf{u}_i}$ such that 
\[
{}^{\ min\ }_{\bar{f},\bar{h}, \bar{g}} \norm{\Koop_i^{T_s}f - \bar{f} - \bar{h} - \bar{g}}_2^2.
\]
Given the data $\DataNet$, we can reformulate the above problem as 
\begin{align}
\label{eq:OptLocal}
\min_{\bar{f},\bar{h}, \bar{g}} \frac{1}{K} \sum_{k = 1}^{K} \norm{\Koop_i^{T_s} f(x_{i,k}) - \bar{f}(x_{i,k}) - \bar{h}(\mathbf{x}_{i,k}) - \bar{g}(\mathbf{u}_{i,k})}_2^2,
\end{align}
where $x_{i,k} \in \Rr^{n_i}$ is the measurement of the $k^{th}$ sample at the node $x_i$. Similarly, $\mathbf{x}_{i,k}$ and $\mathbf{u}_{i,k}$ are the $k^{th}$ samples of the neighbor states and inputs of the node $i$, as in equation \eqref{eq:nei_state_ip}. Defining the matrices 
\begin{align}
\label{eq:LocLiftedMat_Main}
P_{x_i} = \begin{bmatrix} p_i(x_{i,1}) \\ p_i(x_{i,2})\\ \vdots \\ p_i(x_{i,K}) \end{bmatrix}, \ P_{y_i} = \begin{bmatrix} p_i(y_{i,1}) \\ p_i(y_{i,2}) \\ \vdots \\ p_i(y_{i,K}) \end{bmatrix}, \
P_{\mathbf{x}_i} = \begin{bmatrix} \bq_i(\mathbf{x}_{i,1}) \\ \bq_i(\mathbf{x}_{i,2}) \\ \vdots \\ \bq_i(\mathbf{x}_{i,K}) \end{bmatrix}, \
P_{\mathbf{u}_i} = \begin{bmatrix} \br_i(\mathbf{u}_{i,1}) \\ \br_i(\mathbf{u}_{i,2})\\ \vdots \\ \br_i(\mathbf{u}_{i,K}) \end{bmatrix},
\end{align}
we rewrite the optimization problem (\ref{eq:OptLocal}) as 
\[
\min_{\bar{A}_i,\bar{B}_i,\bar{E}_i} \norm{P_{y_i} - P_{x_i} \bar{A}_i^T - P_{\mathbf{x}_i} \bar{E}_i^T - P_{\mathbf{u}_i} \bar{B}_i^T}_2^2,
\]
whose solution is given by
\[
\begin{bmatrix} \bar{A}_i^T \\ \bar{E}_i^T \\ \bar{B}_i^T \end{bmatrix} = [P_{x_i}\ P_{\mathbf{x}_i}\ P_{\mathbf{u}_i}]^{\dagger} P_{y_i}. 
\]
This yields the following discrete-time lifted input-state dynamics
\begin{align}
    \label{eq:LocLiftDyn_DT}
    z_i[t+1] \approx \bA_i z_i[t] + \bE_i w_i[t] + \bB_i {v}_i[t],\qquad t \in \mathbb{Z},
\end{align}
with states $z_i[t] = p_i(x_i(tT_s)) \in \Rr^{N_i}$, ${v}_i[t] = \br_i(\mathbf{u}_i(tT_s))  \in \Rr^{\bar{M}_i}$, and $w_i[t] = \bq_i(\mathbf{x}_i(tT_s)) \in \Rr^{\bar{R}_i}$. Note that the lifted states $w_i$ of the neighboring nodes can be interpreted as external inputs at the node $i$.

\subsubsection{continuous-time local lifted dynamics and estimation of parameters}

Our aim is to identify the matrices $A_i$, $E_i$ and $B_i$ so that the continuous-time dynamics 
\begin{align}
\label{eq:LocLiftDyn_CT}
\dot{z}_i = A_iz_i + E_iw_i + B_i {v}_i,
\end{align}
matches the discrete-time dynamics (\ref{eq:LocLiftDyn_DT}) at sampling time $T_s$, with the same initial condition $z_i(0)$ and the ``external" inputs ${v}_i$, $w_i$ held constant at ${v}_i(0)$ and $w_i(0)$ over the sampling interval $[0,T_s]$. The solution to the dynamics (\ref{eq:LocLiftDyn_CT}) at time $t = T_s$ is given by
\begin{align}
\label{eq:dummy3}
z_i(T_s) = e^{A_iT_s}z_i(0) + \bigg(\int_0^{T_s} e^{A_i(T_s-\tau)}E_i d\tau \bigg) w_i(0) + \bigg(\int_0^{T_s} e^{A_i(T_s-\tau)}B_i d\tau \bigg){v}_i(0).   
\end{align}
Comparing (\ref{eq:LocLiftDyn_DT}) and (\ref{eq:dummy3}), we obtain
\begin{align}
\label{eq:A_mat}
A_i = \frac{1}{T_s} \log(\bA_i).
\end{align}
Moreover, if $A_i$ is invertible, we have
\begin{align*}
\bE_i &= \int_{0}^{T_s} e^{A_i(T_s-\tau)} E_i d\tau = (e^{A_iT_s} - I) A_i^{-1} E_i, \\
\bB_i&= \int_{0}^{T_s} e^{A_i(T_s-\tau)} B_i d\tau = (e^{A_iT_s} - I) A_i^{-1} B_i,
\end{align*}
and
\begin{align}
\label{eq:B_mat}
E_i = A_i(\bA_i-I)^{-1} \bE_i \quad \quad B_i = A_i(\bA_i-I)^{-1} \bB_i.
\end{align}

\subsubsection*{Estimation of parameters}
Decomposing $E_i$ and $B_i$ as
\begin{align}
\label{eq:DecompG}
\begin{aligned}
E_i &= [ E_{i_1} \ E_{i_2} \ \dots \ E_{i_{s_i}}], \\
B_i & = [ B_{i_1} \ B_{i_2} \ \dots \ B_{i_{t_i}}],
\end{aligned}
\end{align}
we have the continuous-time dynamics (\ref{eq:LocLiftDyn_CT}) written as 
\begin{align}
    \dot{z}_i = A_i z_i + \sum_{{i_k} \in \bnei} E_{i_k} w_{i_k} + \sum_{{i_k} \in \bmei} B_{i_k}  v_{i_k},
\label{eq:NetLocDyn}
\end{align}
where $w_{i_k}$ is the lifted state from the neighbor ${i_k}$ defined as $w_{i_k} = q_{i_k}(x_{i_k})$, and $v_{i_k}$ is the lifted input corresponding to the input ${i_k}$ defined as $v_{i_k} = r_{i_k}(u_{i_k})$. Without loss of generality, we assume that the first $n_i$ basis functions $f_{il}$ are the identity functions. The parameters $\alpha_{il}$, $\beta_{ikl}$ and $\gamma_{ikl}$ defined in (\ref{eq:VF_comp}) are then estimated by
\begin{align}
\label{eq:Weights_local}
\begin{aligned}
\alpha_{il} &= A_i[1:n_i,l] \quad 1 \leq l \leq N_i, \\
\beta_{ikl} &= B_{i_k}[1:n_i,l] \quad  1 \leq l \leq M_{i_k},\\
\gamma_{ikl} &= E_{i_k}[1:n_i,l] \quad 1 \leq l \leq P_{i_k},
\end{aligned}
\end{align}
where the vectors $A_i[1:n_i,l]$, $B_{i_k}[1:n_i,l]$, and $E_{i_k}[1:n_i,l]$ contain the first $n_i$ components of the $l^{th}$ column of the matrices $A_i$, $B_{i_k}$, and $E_{i_k}$ respectively.

\subsubsection{Modular Koopman operator for networks} 

As a by product, we can also construct the Koopman operator approximation associated with the complete networked dynamics in a modular form. Consider the vector of all functions acting on $x_i$ and the vector of all functions acting on the input $u_k$, that is,
\begin{eqnarray*}
z_i & = & [f_{i1}(x_i),\dots,f_{iN_i}(x_i), h_{i1}(x_i), \dots, h_{iR_i}(x_i)]^T \in \mathbb{R}^{N_i+R_i}, \\
v_i & = & [g_{i1}(u_i),\dots,g_{iM_i}(u_i)]^T \in \mathbb{R}^{M_i},
\end{eqnarray*}
so that we can construct
\[\textbf{z} = [z_1,\dots,z_N]^T  \quad \quad \textbf{v} = [v_1,\dots,v_M]^T, 
\]
as the lifted states and inputs of the entire network. Given that there is a unique lifted state $z_i$ for each node $x_i$ and that the individual local lifted dynamics of each node is given by \eqref{eq:NetLocDyn}, the lifted dynamics of the entire network can be constructed as
\begin{align} 
\label{eq:EDMD_Network}
\dot{\textbf{z}} = \textbf{A} \textbf{z} + \textbf{B} \textbf{v},
\end{align}
with the matrices
\begin{align*}
\textbf{A} = \begin{bmatrix} A_1 & E_{1,2} & \dots & E_{1,N} \\ E_{2,1} & A_2 & \dots & E_{2,N} \\ \vdots & \vdots & \ddots & \vdots \\ E_{N,1} & E_{N,2} & \dots & A_N \end{bmatrix} &\quad \quad 
\textbf{B} = \begin{bmatrix} B_{1,1} & B_{1,2} & \dots & B_{1,M} \\ B_{2,1} & B_{2,2} & \dots & B_{2,M} \\ \vdots & \vdots & \ddots & \vdots \\ B_{N,1} & B_{N,2} & \dots & B_{N,M}\end{bmatrix}.
\end{align*}
The matrix blocks $E_{i,i_k}$ and $B_{i,i_k}$ are given by
\begin{align*}
E_{i,i_k} = \left\{ \begin{matrix} E_{i_k} & i_k \in \bnei \\ 0 & \mbox{otherwise} \end{matrix} \right. \hspace{.3in} B_{i,i_k} = \left\{ \begin{matrix} B_{i_k} & i_k \in \bmei \\ 0 & \mbox{otherwise} \end{matrix} \right. .
\end{align*}
We note that the data required to construct this modular approximation is much less than the total dimension of the linear lifted network, which is equal to  $\sum_{i=1}^N N_i + R_i$. This provides an efficient and scalable approximation of the Koopman operator for the entire network.

\subsection{Algorithm for Koopman operator based network identification}
The two steps of the method are summarized in Algorithm \ref{alg:first_step} and \ref{alg:sec_step}, respectively. The entire two-step identification method is given in Algorithm \ref{alg:two_steps}.
\begin{algorithm}
\caption{First step: Estimation of Neighbors and Inputs Sets}
\label{alg:first_step}
\begin{algorithmic}[1] 
 \renewcommand{\algorithmicrequire}{\textbf{Input:}}
 \renewcommand{\algorithmicensure}{\textbf{Output:}}
 \REQUIRE Data samples $\DataNet$ and their time derivatives $\dot{X}_k$ computed through the dual method, node functions $\phi_{kl}$ and $\psi_{kl}$, threshold $\delta$. \ENSURE The estimates $\bnei$ and $\bmei$ of the sets $\Nei$ and $\Mei$.
\STATE Using the node functions $\phi_{kl}$ and $\psi_{kl}$, construct the matrix $H_{X,U}$ as in equation \eqref{eq:Hxu-network}. 
\STATE Solve the Lasso problem defined in \eqref{eq:LassoNet} and estimate $\hat{\xi}_i^{(j)}$ for each $i=1,\dots,N$ and $j=1,\dots,n_i$.
\STATE Estimate the weights $\Lambda_{ik}$ and $\Delta_{ik}$ as defined in \eqref{eq:NeiWeights}.
\STATE The estimates of $\Nei$ and $\Mei$ are computed as $\bnei$ and $\bmei$ using the threshold $\delta$ as in equation \eqref{eq:estim_neighbors}.
\end{algorithmic}
\label{alg:Firststep}
 \end{algorithm}

\begin{algorithm}
\caption{Second step: Identification of Local Node Dynamics}
\label{alg:sec_step}
\begin{algorithmic}[1] 
 \renewcommand{\algorithmicrequire}{\textbf{Input:}}
 \renewcommand{\algorithmicensure}{\textbf{Output:}}
 \REQUIRE Data samples $\DataNetAug$, sampling time $T_s$, the estimated neighbor set $\bnei$, input set $\bmei$, and the set of functions $f_{il}(x_i)$, $g_{kl}(u_k)$ for $k \in \bmei$, $h_{kl}(x_k)$ for $k \in \bnei$. \ENSURE The parameters $\alpha_{il}$,$\beta_{kil}$ and $\gamma_{kil}$ corresponding to equation (\ref{eq:VF_comp}).
%
\STATE Define $\textbf{x}_i$ and $\textbf{u}_i$ as 
\[
\textbf{x}_i = \begin{bmatrix} x_{j_1} \\ x_{j_2} \\ \vdots \\ x_{j_{r_1}}\end{bmatrix} \quad \quad j_l \in \bnei, \ r_1 = |\bnei|, \quad \quad \textbf{u}_i = \begin{bmatrix} u_{j_1} \\ u_{j_2} \\ \vdots \\ u_{j_{r_2}}\end{bmatrix} \quad \quad j_l \in \bmei,\ r_2 = |\bmei|.
\]
\STATE Construct the matrices $P_{x,i}$,$P_{y,i}$, $P_{w,i}$ and $P_{v,i}$ as in equation (\ref{eq:LocLiftedMat_Main}).
\STATE Compute $\bA_i$, $\bE_i$ and $\bB_i$ as in equation (\ref{eq:LocLiftDyn_DT}).
\STATE Construct the matrix $A_i$,$E_i$ and $B_i$ corresponding to the continuous-time lifted dynamics (\ref{eq:LocLiftDyn_CT}) using the equations (\ref{eq:A_mat}) and (\ref{eq:B_mat}). 
\STATE Decompose the matrices $E_i$ and $B_i$ as in equation (\ref{eq:DecompG}) and construct matrices $E_{i_k}$, $i_k \in \bnei$ and $B_{i_k}$, $i_k \in \bmei$. 
\STATE The parameters $\alpha_{il}, \beta_{kil}$ and $\gamma_{kil}$ are computed as in equation (\ref{eq:Weights_local}). 
\end{algorithmic}
\label{alg:Local_ip}
 \end{algorithm}

\begin{algorithm}
\caption{Two-Step Network Identification}
\label{alg:two_steps}
\begin{algorithmic}[1] 
 \renewcommand{\algorithmicrequire}{\textbf{Input:}}
 \renewcommand{\algorithmicensure}{\textbf{Output:}}
 \REQUIRE Data samples $\DataNet$, sampling time $T_s$, set of functions $f_{il}$, $g_{kl}$, $h_{kl}$ according to equation (\ref{eq:VF_comp}), node functions $\phi_{kl}$ and $\psi_{kl}$ and threshold $\delta$. 
 \ENSURE The parameters $\alpha_{il}$,$\beta_{il}$ and $\gamma_{jil}$ corresponding to equation (\ref{eq:VF_comp}) for each node $i$.
%
\STATE From the data, compute the approximation of the Koopman operator over the ``sample space" by using equation (\ref{eq:State_approx_dual}). 
\STATE Compute the estimates of the vector field $\dot{X}_k$ at each sample point $(X_k,U_k)$. 
\STATE Use Algorithm \ref{alg:Firststep} to estimate $\bnei$ and $\bmei$ for each $i=1,\dots,N$.
\FOR{i = 1,\dots,N}
\STATE Estimate the parameters $\alpha_{il}$, $\beta_{kil}$ and $\gamma_{kil}$ using Algorithm \ref{alg:Local_ip}. 
\ENDFOR
\end{algorithmic}
\label{alg:Twostep}
\end{algorithm}

\section{Numerical Examples}
\label{sec:NE}

In this section, we provide numerical examples to support the framework developed in this paper. In particular, we investigate the performance of the  identification algorithm in three cases:
\begin{enumerate}
    \item \textbf{Dynamic network defined on an Erdos-Renyi graph with polynomial coupling functions:} We study the identification error with respect to network size, sparsity, number of measurements and intensity of measurement noise. Our results are compared  with the dual method proposed in \cite{MG2020}. We also reconstruct the Boolean network of interactions, which is obtained through the first step of our method.
       
    \item \textbf{Sparse network with non-polynomial coupling functions.} We study the identification error with respect to the sampling time $T_s$ in data-generation. Additionally, we focus on the local identification of a random node to illustrate the accuracy of the proposed framework. 
    \item \textbf{Bio-inspired network with non-scalar local dynamics (Hindmarsh-Rose neuron) and non-polynomial coupling functions.} We investigate the effect of the threshold $\delta$ on the Boolean reconstruction of the network and the accuracy of the method in the local identification.
\end{enumerate} 

The following error metrics will be used to evaluate the identification performance:
\begin{itemize}
       \item \textbf{Local Error $\epsilon_i$} at each node $i$ defined as 
       \begin{align}
       \label{eq:local_error_epsilon}
       \epsilon_i^2 = \sum_{k \in \Nei} \sum_{l} (\beta_{kil} - \hat{\beta}_{kil})^2 + \sum_{k \in \Mei} \sum_{l} (\gamma_{kil} - \hat{\gamma}_{kil})^2.
    \end{align}
    \item \textbf{Root Mean Square Error (RMSE)} which provides an overall performance of the identification method and defined as 
    \[
    RMSE = \sqrt{\frac{1}{N} \sum_{i=1}^{N} {\epsilon_i}^2}.
    \]
    \item \textbf{Maximum Error (ME)} defined as 
    \[
    ME = {}^{\ \ \ max\ \ \ }_{i = \{1,\dots,N\}}\ \  \epsilon_i,
    \]
    which quantifies the worst case performance in the local identification of a node.
\end{itemize}

Apart from the above metrics, we also consider the True Positive Rate (TPR) and False Positive Rate (FPR) in the identification of the Boolean network of interactions, which are defined as follows:

   \begin{align*}
    \mbox{TPR} &=\frac{\mbox{Number of correctly identified edges}}{\mbox{Number of edges}}, \\
    \mbox{FPR} &= \frac{\mbox{Number of falsely identified edges}}{\mbox{Number of non existing edges}},
\end{align*}
 where the number of non-existing edges is the difference between the number of edges in a complete graph of $N$ nodes and the total number of edges in the network. We will focus on the Receiver Operating Characteristic (ROC) curve plotting the True Positive Rate (TPR) against the False Positive Rate (FPR) as the threshold $\delta$ is varied (see Section \ref{sec:estimate_neighbors}). In particular, the area under the ROC curve (AUROC) is a metric which can be used to estimate the effectiveness of the binary identification.

\subsection{Erdos-Renyi network}

We consider a directed Erdos-Renyi network of $N$ nodes, with probability $\rho$ for the existence of an edge between any pair of nodes. The total number of external inputs is assumed to be two. The network dynamics is given by
\[
\dot{x}_i = \sum_{k \in \Nei} \big( \beta_{ki1} x_k + \beta_{ki2} x_k^2 + \beta_{ki3} x_k^3 \big) + \sum_{k \in \Mei} \big( \gamma_{ki1} u_k + \gamma_{ki2} u_k^2 \big).
\]
The functions $H_{ik}$ are monomials of degree uniformly randomly distributed over the set $\{1,2,3\}$ and the input functions $G_{ik}$ are also monomials of degree uniformly distributed over the set $\{1,2\}$, which implies that one of $\beta_{ki1},\beta_{ki2},\beta_{ki3}$ is non-zero for $k\in \Nei$ and one of $\gamma_{ki1},\gamma_{ki2}$ is non-zero for $k \in \Mei$. The data is generated for $K$ trials with initial conditions $X_k \in \Rr^N$ and inputs $U_k \in \Rr^2$ uniformly randomly distributed over the set $[-1,1]^{N}$ and $[-1,1]^2$, and the values $Y_k$ are generated with the sampling time $T_s = 0.1$. Additionally, two independent white Gaussian noise vectors $v_1,v_2\ \ \in \Rr^N$ with zero mean and variance $\sigma$ are added to the data $X_k$ and $Y_k$. 

We perform the first step of the method, computing an estimate of the vector field through the dual identification method, with $2K$ Gaussian RBF centered around the data points $\bar{X}_k,\bar{Y}_k$, with $\gamma=0.001$. Next, the estimated sets of neighbors $\bar{\Nei}$ and inputs $\bar{\Mei}$ are computed with the node functions $\cup_{i=1}^{N}\{x_i,x_i^2\}$ and $\cup_{i=1}^2 \{u_i,u_i^2\}$ and with a threshold $\delta$ (see \eqref{eq:estim_neighbors}). Finally, local identification is performed with the main method for each $i$ in order to recover the functions $H_{ik}$ and $G_{ik}$. The local dictionary functions are taken to be monomials up to degree $4$ in $x_i$ (i.e. $\{x_i,x_i^2,x_i^3,x_i^4\}$) and up to degree $2$ in $u_i$ (i.e. $\{u_i,u_i^2\}$).

For the Erdos-Renyi graph, we consider the following experiments to illustrate the proposed two-step network identification.
\begin{enumerate}
    \item \textbf{Varying the number of data points $K$:} As expected, the identification error decreases as the number of data point increases (Figure \ref{fig:Effect_K}).
    \item \textbf{Varying the size of the network $N$ (scalability):} In this experiment, we fix $K = 2N$. Our two-step method scales up better with network size with consistent performance on both RSME and ME, compared with the dual method (Figure \ref{fig:Effect_N}). 
    \item \textbf{Varying the variance of measurement noise $\sigma$:}

    We observe that the two-step method has lower error consistently over the dual method for different noise variances (Figure \ref{fig:Effect_sigma}). Also, as the noise variance decreases, the two-step method shows an improvement in identification error while the dual method has consistently  large errors even for a small noise variance.
    \item \textbf{Varying the probability $\rho$ of an edge between two nodes (sparsity):} Given the same amount of data, we can see that the identification error increases as sparsity decreases, but the increase is steeper with the dual method for low sparsity (Figure \ref{fig:Effect_rho}).
    \item \textbf{Reconstruction of the Boolean network:} The network reconstruction is obtained with the first step of the method, by using the predicted neighbor set $\bnei$ obtained with a threshold $\delta$ (defined in equation \eqref{eq:estim_neighbors}). The results are excellent for low sparsity of the network. Moreover the AUC decreases as the sparsity decreases (see Figure \ref{fig:AUC}(a) for $\rho = 0.25,0.275,0.3$ and $K = 250$). A possible solution to counter this effect is to increase the number of data samples $K$ (see  Figure \ref{fig:AUC}(a) for $\rho = 0.3$ and $K = 350$). Figure \ref{fig:AUC}(b) shows the comparison of ROC curve for the two-step and dual method with the same dataset and it can be seen that two-step method outperforms the dual method by a considerable margin.
\end{enumerate}

\begin{figure}
    \centering
    \begin{subfigure}{0.49\textwidth}
        \centering
        \includegraphics[width = \textwidth]{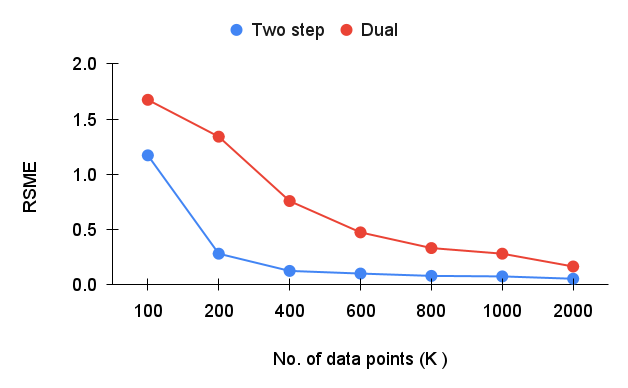}
        \caption{RSME}
    \end{subfigure} 
    \begin{subfigure}{0.49\textwidth}
        \centering
        \includegraphics[width = \textwidth]{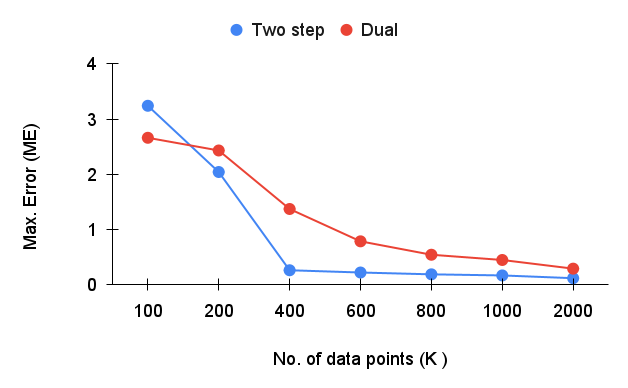}
        \caption{Maximum Error}
    \end{subfigure}
    \caption{Effect of varying the amount of data on a network with $N = 400$, probability of an edge $\rho = 0.005$, noise variance $\sigma = 0.01$, and threshold $\delta = 0.1$.}
    \label{fig:Effect_K}
\end{figure}
    \begin{figure}
    \centering
    \begin{subfigure}{0.49\textwidth}
        \centering
        \includegraphics[width = \textwidth]{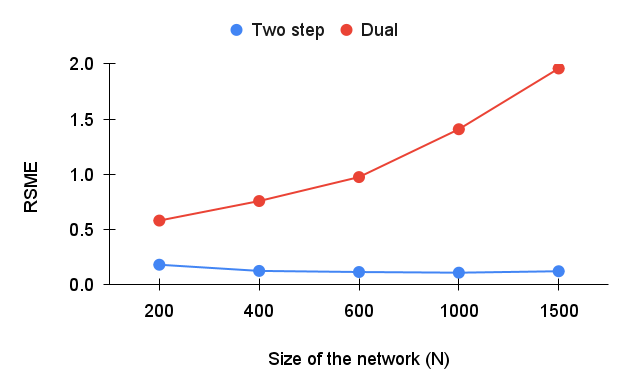}
        \caption{RSME}
    \end{subfigure} 
    \begin{subfigure}{0.49\textwidth}
        \centering
        \includegraphics[width = \textwidth]{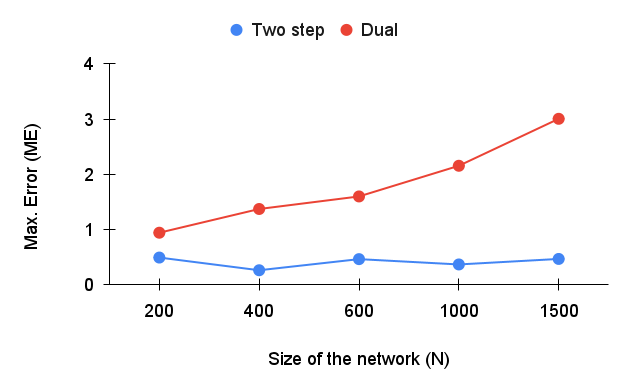}
        \caption{Maximum Error}
    \end{subfigure}
    \caption{Effect of varying the size of a network with data points $K = N$, probability of an edge $\rho = 0.005$, noise variance $\sigma = 0.01$, and threshold $\delta = 0.1$.}
    \label{fig:Effect_N}
\end{figure}
\begin{figure}
    \centering
    \begin{subfigure}{0.49\textwidth}
        \centering
        \includegraphics[width = \textwidth]{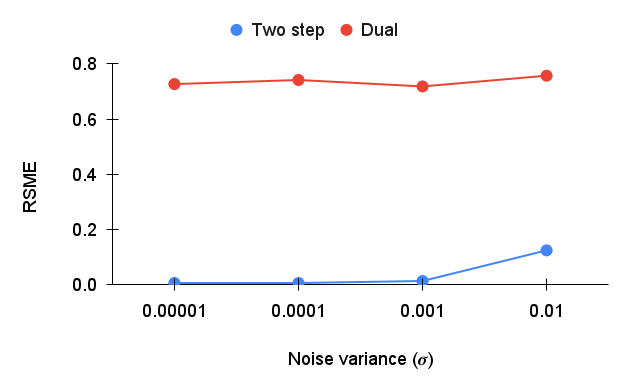}
        \caption{RSME}
    \end{subfigure} 
    \begin{subfigure}{0.49\textwidth}
        \centering
        \includegraphics[width = \textwidth]{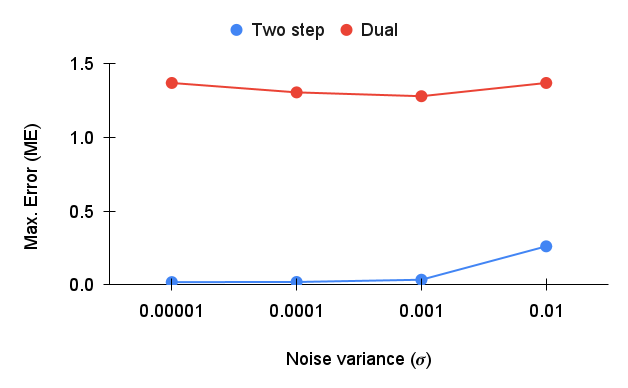}
        \caption{Maximum Error}
    \end{subfigure}
    \caption{Effect of varying the noise variance on a network with $N = 400$, probability of an edge $\rho = 0.005$, number of data points $K = 400$, and threshold $\delta = 0.1$.}
    \label{fig:Effect_sigma}
\end{figure}
\begin{figure}
    \centering
    \begin{subfigure}{0.49\textwidth}
        \centering
        \includegraphics[width = \textwidth]{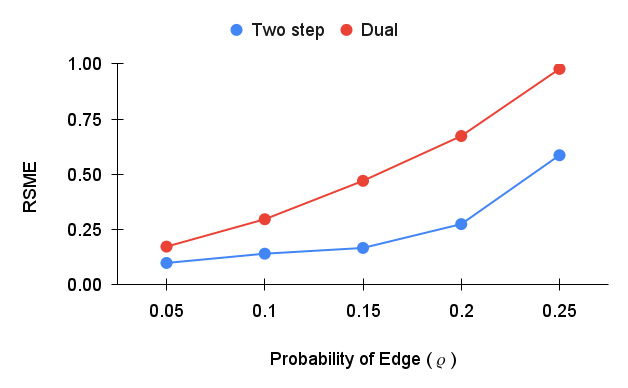}
        \caption{RSME}
    \end{subfigure} 
    \begin{subfigure}{0.49\textwidth}
        \centering
        \includegraphics[width = \textwidth]{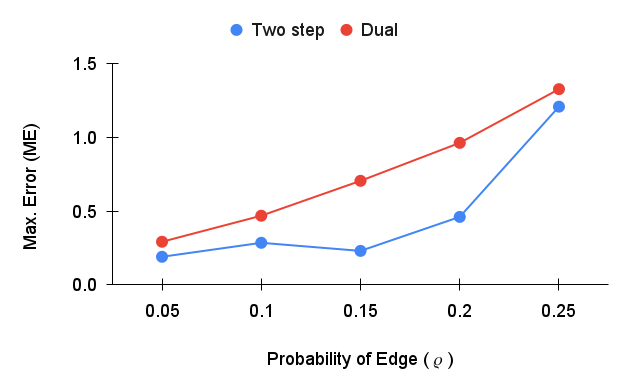}
        \caption{Maximum Error}
    \end{subfigure}
    
    \caption{Effect of varying the edge probability $\rho$ with $N = 50$, noise variance $\sigma = 0.01$, and number of data points $K = 600$.}
    \label{fig:Effect_rho}
\end{figure}
\begin{figure}
    \centering
    \begin{subfigure}{0.49\textwidth}
        \centering
        \includegraphics[width = \textwidth]{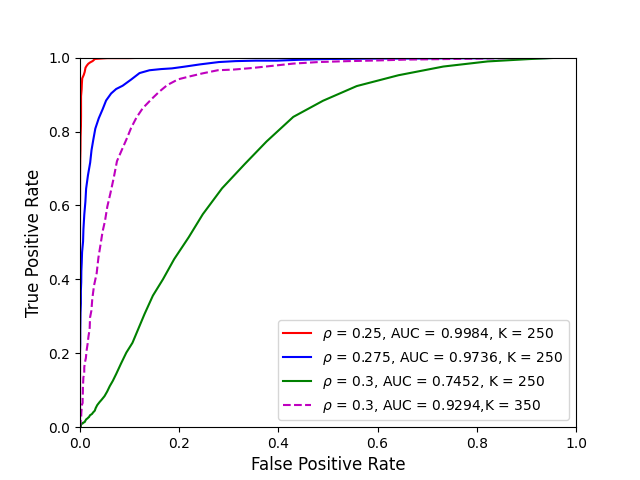}
        \caption{ROC curves for several edge probability values $\rho$}
    \end{subfigure} 
    \begin{subfigure}{0.49\textwidth}
        \centering
        \includegraphics[width = \textwidth]{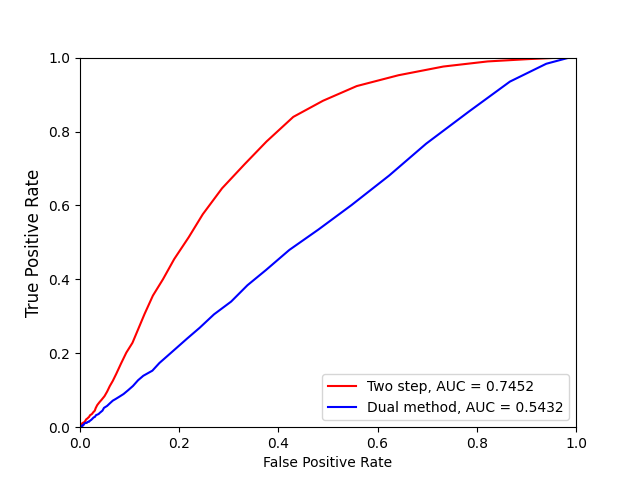}
        \caption{ROC comparison with $\rho = 0.3$ and $K = 250$}
    \end{subfigure}
    
    \caption{ROC curve for the Boolean reconstruction of the Erdos-Renyi network with $N = 75$ nodes.}
    \label{fig:AUC}
\end{figure}

\subsection{Network with a non-polynomial vector field}

We consider a network with $N = 200$ nodes and $4$ inputs, and the dynamics
   \begin{align}
                \dot{x}_i &= \left\{ \begin{matrix} -0.5 x_i^2 -0.5 x_{47i\ \mbox{mod}\ \mbox{n}}+0.7 x_{(i+1)\ \mbox{mod}\ \mbox{n}} - 0.5 \ sin(x_{t_i})  + 1.4 u_1  & \mbox{if}\ (i-1)\ \mbox{mod}\ 4 = 0,\vspace{.1in}\\ 
                 -0.5 x_i + 0.7 x_{i-1}^2 + 0.7 x_{23i\ \mbox{mod}\ \mbox{n}}^3 + 0.7 e^{x_{t_i}} + 1.4 u_4^2  & \mbox{if}\ (i-1)\ \mbox{mod}\ 4 = 1,\vspace{.1in} \\
                 -0.5 x_i + 0.7 x_{(i+1)\ \mbox{mod}\  \mbox{n}}^2 - 0.5 x_{67i\ \mbox{mod}\ \mbox{n}} + 0.5 e^{x_{t_i}} + 1.4 u_2^2  & \mbox{if}\ (i-1)\ \mbox{mod}\ 4 = 2,\vspace{.1in} \\
                 -0.5 x_i^2 - 0.5 x_{i-1}^2 + 0.7 x_{11i\ \mbox{mod} \ \mbox{n}}^3 - 0.5 sin(x_{t_i}) + 1.4 u_3^2 & \mbox{otherwise},\end{matrix} \right.
    \end{align}
    
     where the additional neighbor $x_{t_i}$ of each node $i$ is assigned randomly from the set of all nodes. The data is generated with $K=300$ different initial conditions $X_i \in \mathbb{R}^{200}$ and $U_i\in \Rr^4$ uniformly randomly distributed over $[-1,1]^{N}$ and $[-1,1]^4$, respectively. For the estimation of the vector field using the dual identification method, we use $300$ Gaussian RBF centered around $\bar{X}_i$. The node functions for estimating the neighbors are chosen to be monomials up to degree $2$ on the nodes $x_i$ and inputs $u_i$. For the local identification, the functions $\{x_i,x_i^2,x_i^3,\sin(x_i),e^{x_i}\}$ and $\{u_i,u_i^2\}$ are considered as the dictionary functions. Figure \ref{fig:NonPolyVF} shows the estimated and true parameters at node $100$ for the sampling times $T_s = 0.01,0.05$, and $0.1$. The RSME and the Maximum Error (ME) are given in Table \ref{tab1}. As expected, the local identification error increases as $T_s$ increases.     
    
    \begin{figure}
     \centering    
    \begin{subfigure}{0.45\textwidth}
        \centering
        \includegraphics[width = \textwidth]{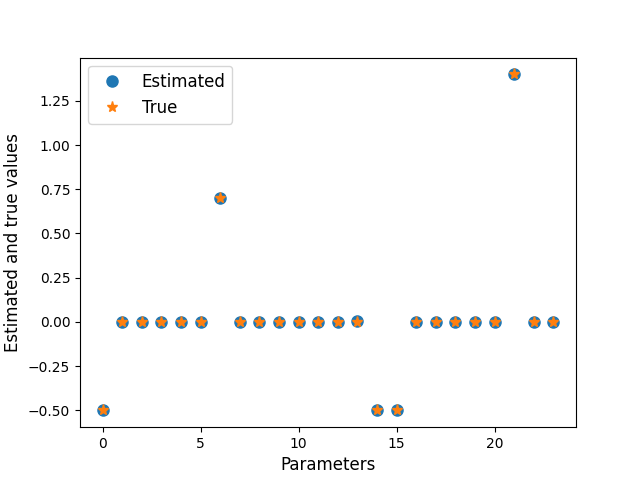}
        \caption{$T_s = 0.01$}
    \end{subfigure}
    \begin{subfigure}{0.45\textwidth}
        \centering
        \includegraphics[width = \textwidth]{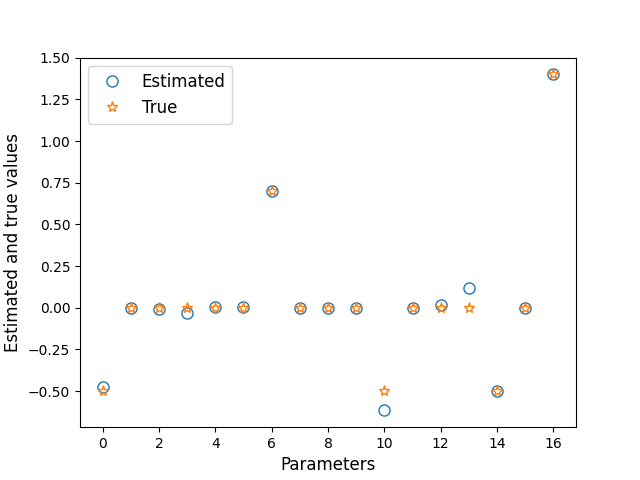}
        \caption{$T_s = 0.05$}
    \end{subfigure}   
    \begin{subfigure}{0.45\textwidth}
        \centering
        \includegraphics[width = \textwidth]{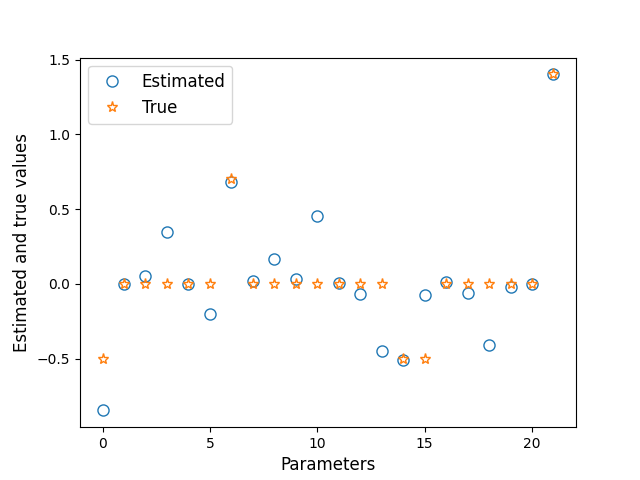}
        \caption{$T_s = 0.1$}
    \end{subfigure}
    \caption{Estimation of the parameters in a network with a non-polynomial vector field at node 50, for different sampling times.}
    \label{fig:NonPolyVF}
    \end{figure}
   
\begin{table}[ht]
\begin{center}
    \begin{tabular}{| c | c | c |}
  \hline			
  \mbox{Sampling time} & \mbox{RSME} & \mbox{Max Error (ME)} \\
  \hline 
  0.01 & 0.016 & 0.049 \\
  \hline 
  0.05 & 0.45 & 1.72 \\
  \hline  
  0.1 & 1.86 & 6.13\\
  \hline 
\end{tabular}
\end{center}
\caption{RSME and maximum error (ME) for the identification of a network with a non-polynomial vector field.}
\label{tab1}
\end{table}

\subsection{Identification of networks with non-scalar local Hindmarsh-Rose dynamics}

We now consider a bio-inspired small world Watts-Strogatz network \cite{BDS_2017} containing $N=75$ nodes with $8$ neighbors in average and a measure of randomness $\beta = 0.5$. The local dynamics is given by the Hindmarsh-Rose (HR) neuron model \cite{IEM2004} described by
\begin{align}
\begin{aligned}
\dot{x}_i &= y_i - b_ix_i^2 +a_i x_i^3 - z_i + C_i, \\
\dot{y}_i &= c_i - d_ix_i^2 - y_i, \\
\dot{z}_i &= \frac{1}{\tau}(s_i\,(x_i -e_i) - z_i),
\end{aligned}
\end{align}
with the additional nonlinear coupling
\begin{align}
C_i = \sum_{j \in \Nei} \frac{4}{1+e^{-\nu_j(x_j - \theta_{ij})}}.
\end{align}
The parameters of the local dynamics are $a_i \in \{1,1.25,1.5,1.75,2\}$, $b_i \in \{2,2.75,3.5,$ $4.25,5\}$, $d_i \in \{ -3,-3.5,-4,-4.5,-5\}$, $s_i \in \{8,11,14,17,20 \}$, and $e_i \in \{ -4,-2,0,2,4\}$. The parameters of the coupling functions are $\tau = 1000$, $\nu_j = 1$, and $\theta_{ij} \in \{-0.5,-1,-1.5\}$. Our goal is to compute the coefficients $a_i,b_i,c_i,d_i,e_i,s_i$, the set $\Nei$ for all $i$, and the parameters $\theta_{ij}$ for $j \in \Nei$.
We generate $500$ initial conditions uniformly distributed over $[-1,1]^{225}$ and use a sampling time $T_s = 0.01$. For the estimation of the vector field using the dual identification method, we use $500$ Gaussian RBFs centered around data points ${X}_k$. In the first step, we use the functions $x_i, x_i^2,y_i,z_i$ (i.e. $300$ node functions) to identify the Boolean network. In the second step, we use $x_i,x_i^2,x_i^3,y_i,z_i,1$ as the dictionary functions for the local states
and
\[
\frac{1}{1+e^{-(x_j + 0.5)}},  \frac{1}{1+e^{-(x_j + 1)}},\frac{1}{1+e^{-(x_j + 1.5)}},
\]
for the neighbors, which leads to $3|\bnei|+6$ dictionary functions in the local identification and a total of $3|\bnei|+18$ parameters (as we consider $6$ local dictionary functions for each local state $x_i,y_i$ and $z_i$, and an additional $3|\bnei|$ dictionary functions for the coupling). 

The Boolean network reconstruction obtained with the first step of the method is shown in Figure \ref{fig:HR_Network3} for different threshold values. It can be seen that the method provides excellent reconstruction of the network for intermediate threshold values ($\delta=0.4$). Table \ref{tab:HR_Network1} gives the statistical properties of the local error $\epsilon_i$, as defined in \eqref{eq:local_error_epsilon} for the identification of each node. The true vs estimated parameters are shown in Figure \ref{fig:HR_Network2} for the node $i$ corresponding to the maximum local error (ME). For this node, the total number of estimated neighbors $\bnei$ is $16$ while the true number of neighbors is $7$, leading to the identification of an unnecessarily large number of parameters ($66$ in total).

\begin{figure}
    \centering
    \begin{subfigure}{0.49\textwidth}
        \centering
        \includegraphics[width = \textwidth]{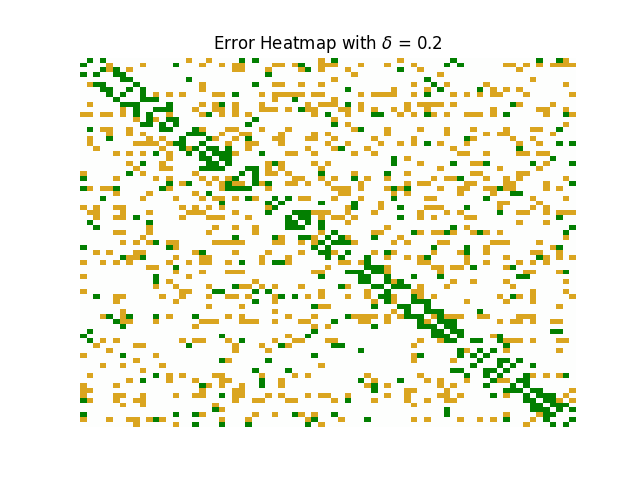}
    \end{subfigure} 
    \begin{subfigure}{0.49\textwidth}
        \centering
        \includegraphics[width = \textwidth]{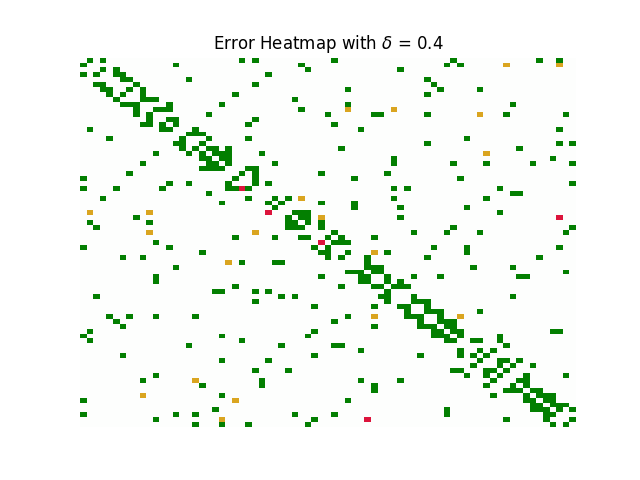}
    \end{subfigure}
    \begin{subfigure}{0.49\textwidth}
        \centering
        \includegraphics[width = \textwidth]{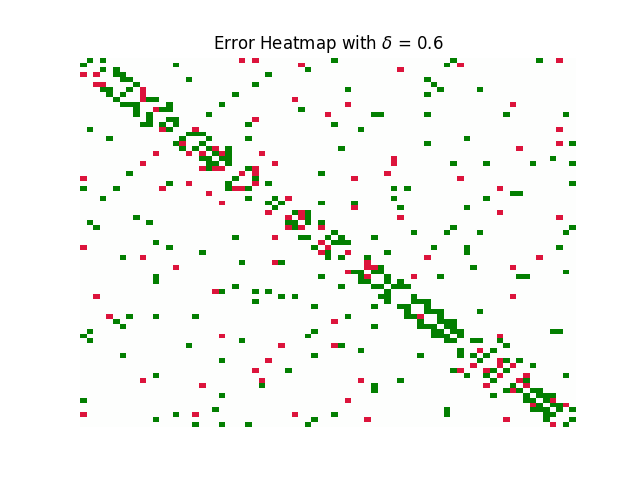}
    \end{subfigure}
    \caption{Boolean reconstruction Watts-Strogatz network for different threshold values $\delta$. Green, yellow, and red dots correspond to correctly identified, false positive, and false negative edges, respectively.}
     \label{fig:HR_Network3}
\end{figure}

\begin{table}[ht]
    \centering
 \begin{tabular}{| c | c |}
  \hline	
  \mbox{Attritbute} & \mbox{Value} \\
  \hline\hline
  \mbox{RSME} & 0.335 \\
  \hline 
  \mbox{Maximum Error} & 1.01 \\
  \hline 
  \mbox{Minimum Error} & 0.056 \\
  \hline  
  \mbox{Standard Deviation} & 0.205\\
  \hline 
\end{tabular} 
         \caption{Statistical properties of local error $\epsilon_i$ for Watts-Strogatz Network  with HR neuron local dynamics}
         \label{tab:HR_Network1}
    \end{table}
    \begin{figure}
        \centering
        \includegraphics[width = .75\textwidth]{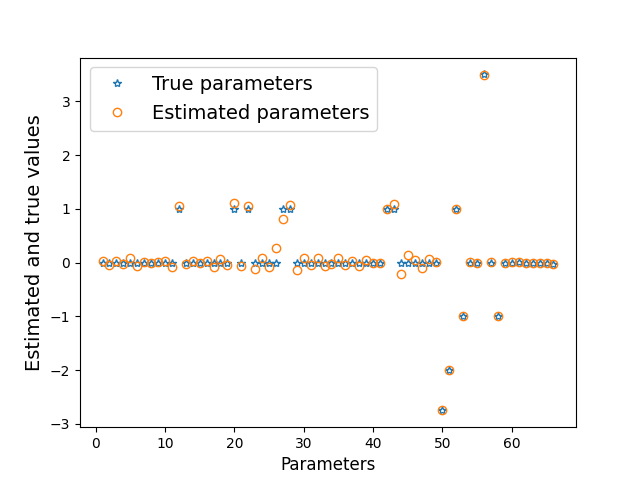}
        \caption{Predicted vs true values of parameters at node $50$}
     \label{fig:HR_Network2}
    \caption{Local identification of a Watts-Strogatz network with Hindmarsh-Rose dynamics.}
\end{figure}

\section{Conclusion}
\label{sec:conclu}

We have developed a Koopman operator based approach to identify nonlinear networked dynamics with external inputs. From nodes and input data, we have estimated the network topology and recovered the (possibly different) local dynamics attached to each node. 
Numerical simulations demonstrate that the proposed identification procedure scales up to large networks without much compromise on the error in identification.

As a by-product to our network identification procedure, we have obtained a linear lifted model for both the local and global dynamics of the network. This provides a modular Koopman operator representation of the nonlinear network dynamics, which could be used for analysis and control design. One prospective application in this direction would be to develop data-driven simulation models of complex nonlinear networks arising in engineering applications. An immediate extension is to apply this framework to models arising from Neural Networks, as well as networks with time-varying topology. Besides identification, investigations on the consensus of nonlinear networks through this framework could be another promising future direction to pursue. 


\bibliographystyle{spmpsci}      
\bibliography{mybib.bib}   

\end{document}